\documentclass[twocolumn]{autart}    

\usepackage{graphicx}          
                               
\usepackage{verbatim}
\usepackage{titletoc}
\usepackage{caption}
\usepackage{subfigure}
\usepackage{algorithm}  
\usepackage{enumerate}
\usepackage{algorithmicx}
\usepackage{algpseudocode} 
\usepackage[most]{tcolorbox}
\usepackage{wrapfig}
\usepackage{subfigure}
\usepackage{bm} 
\usepackage{amsfonts}
\usepackage{amsmath} 
\usepackage{amssymb}
\usepackage{array}
\usepackage{xcolor}
\usepackage{comment}
\usepackage{longtable}
\usepackage{multirow}
\usepackage{multicol}
\usepackage{fancyhdr}
\usepackage{setspace}
\usepackage[numbers,sort&compress]{natbib}
\usepackage{booktabs}
\usepackage{cuted}

\allowdisplaybreaks[4]
\newcommand{\tr}{{\rm Tr}}

\begin{document}

\begin{frontmatter}

\title{Harvest and Jam: Optimal Self-Sustainable Jamming Attacks against Remote State Estimation
}


\author[HKUST]{Yuxing Zhong}
\ead{yuxing.zhong@connect.ust.hk, yuxing.zhong@sydney.edu.au}\thanks{Yuxing Zhong was with the Department of Electronic and Computer Engineering, the Hong Kong University of Science and Technology, Hong Kong, when this work was conducted, and is now with the School of Electrical and Computer Engineering, the University of Sydney, Australia.},    
\author[NU]{Yuzhe Li}\ead{yuzheli@mail.neu.edu},  
\author[USyd]{Daniel E.~Quevedo}\ead{daniel.quevedo@sydney.edu.au},  
\author[HKUST,Bio]{Ling Shi}\ead{eesling@ust.hk}  

\address[HKUST]{Department of Electronic and Computer Engineering, the Hong Kong University of Science and Technology,\\ Clear Water Bay, Kowloon, Hong Kong}
\address[NU]{State Key Laboratory of Synthetical Automation for Process Industries, Northeastern University, Shenyang 110819, China} 
\address[USyd]{School of Electrical and Computer Engineering, the University of Sydney, Sydney, Australia} 
\address[Bio]{Department of Chemical and Biological Engineering, Hong Kong University of Science and Technology,\\ Clear Water Bay, Kowloon, Hong Kong}  
          
\begin{keyword}                           
Attack, energy harvesting, Markov decision process, power control, and state estimation.          
\end{keyword}                             

\begin{abstract}                          
This paper considers the optimal power allocation of a jamming attacker against remote state estimation. The attacker is self-sustainable and can harvest energy from the environment to launch attacks. The objective is to carefully allocate its attack power to maximize the estimation error at the fusion center. Regarding the attacker's knowledge of the system, two cases are discussed: (i) perfect channel knowledge and (ii) unknown channel model. For both cases, we formulate the problem as a Markov decision process (MDP) and prove the existence of an optimal deterministic and stationary policy. Moreover, for both cases, we develop algorithms to compute the allocation policy and demonstrate that the proposed algorithms for both cases converge to the optimal policy as time goes to infinity. Additionally, the optimal policy exhibits certain structural properties that can be leveraged to accelerate both algorithms. Numerical examples are given to illustrate the main results.
\end{abstract}
\end{frontmatter}

\section{Introduction}\label{chap:intro}
The rapid development of the Internet of Things has led to the implementation of cyber-physical systems (CPSs) across various sectors such as smart grids~\cite{karnouskos2011cyber} and health care~\cite{gunes2014survey}. Nevertheless, the utilization of CPSs presents a dual nature. While enhancing system flexibility, it also makes them vulnerable to malicious attacks~\cite{pasqualetti2015control}. For instance, the power outage in western Ukraine in late 2015 was reported to be caused by a false data injection attack~\cite{foxbrewster2016ukraine}. Generally speaking, attacks can be classified into three categories: deception~\cite{guo2016optimal}, disclosure~\cite{leong2018transmission1}, and disruption~\cite{7054460}. For a comprehensive overview of these attack types, readers are referred to~\cite{sanchez2019bibliographical} and the references therein. This paper focuses on disruption attacks, specifically jamming attacks, where an attacker deliberately interferes with the communication channel. Such attacks result in packet loss, thereby compromising system stability and degrading estimation accuracy.

In systems with a single dynamic process, prior work has addressed the derivation of the optimal attack strategy for the single-sensor case~\cite{7054460} and has extended this analysis to signal-to-interference-plus-noise ratio (SINR)-based networks~\cite{qin2020optimal}. Subsequent studies have generalized these results to systems with multiple sensors~\cite{wang2024optimal}. For systems involving multiple processes, research has investigated optimal attack allocation~\cite{peng2017optimal} and power allocation~\cite{ren2018attack} strategies, as well as attack power design using game-theoretic frameworks~\cite{li2016sinr}. However, the aforementioned works rely on the strong assumption that the attacker possesses complete knowledge of the system. To address this limitation, more recent efforts have considered scenarios where attackers operate with incomplete system information and have utilized reinforcement learning techniques to design effective attacks~\cite{huang2022learning}.

Despite the above extensive body of work, one key limitation remains: most studies assume that attackers operate under strict energy constraints. Recent advancements in energy-harvesting technologies, which enable devices to harvest energy from environmental sources such as solar power and electromagnetic radiation~\cite{ho2012optimal,li2016power}, have the potential to alleviate these constraints significantly. In fact, attackers equipped with energy harvesters can continuously replenish their energy reserves, enabling sustained and potentially more devastating jamming attacks. Surprisingly, the implications of such self-sustaining ``harvest-and-jam" attacks remain unexplored in the existing literature. This gap motivates the present study, which investigates the impact of the ``harvest-and-jam" attacks on system performance. It turns out that modeling energy harvesting as a stochastic process~\cite{ho2012optimal} introduces significant analytical challenges compared to conventional attack frameworks. In summary, the main contributions of this paper are:\\
(1) \underline{\it Energy Havester}: Unlike prior studies where the attacker operates solely on its own energy reserves~\cite{peng2017optimal,ren2018attack,li2016sinr,huang2022learning}, we study a novel scenario where the attacker can harvest energy from the environment. \\
(2) \underline{\it MDP Formulation}: We address the problem using MDP and prove the existence of a deterministic and stationary optimal policy ({\bf Theorem~\ref{thm:stationary}}). {It is important to note that, since the reward function in our setting is not upper bounded, standard results~\cite[Theorem 5.5.4 and Assumption 4.2.1]{hernandez2012discrete} are not directly applicable. Consequently, proving the existence of such a policy becomes significantly more involved~\cite{guo2006average}. This stands in stark contrast to the energy-harvesting sensor case~\cite{leong2018optimal, leong2018transmission}, where standard results can be applied.}\\
(3) \underline{\it Channel Information}: Unlike the previous studies that assume the attacker enjoys complete system information~\cite{peng2017optimal,ren2018attack,li2016sinr}, we consider two different scenarios: (i) perfect channel knowledge and (ii) unknown channel models. This broadens the applicability of our results.\\
(4) \underline{\it  Structural Results}: We design efficient algorithms {with guaranteed convergence} to compute optimal attack policies for both perfect and unknown channel scenarios. We further establish several key structural properties of the optimal policy ({\bf Lemma~\ref{lemma:monototic q}}, {\bf Lemma~\ref{lemma:super}} and {\bf Theorem~\ref{thm:structure}}), which can be utilized to accelerate computation. While developing algorithms to compute optimal policies is common, our work uniquely applies these structural insights to the specific problem. {More importantly, for the unknown channel scenario, we propose a primal-dual-based structural update rule. By explicitly embedding the structure into the optimization constraints, our approach builds a rigorous bridge between stochastic approximation and constrained optimization. This differs from the existing literature, which mainly uses these structural results to design threshold-seeking algorithms~\cite{yang2022joint,nourian2014optimal}, where convergence to the optimal solution cannot be guaranteed.}

\noindent\emph{Notations}: The notations $\mathbb{R}$, $\mathbb{R}^n$ and $\mathbb{R}^{n\times m}$ denote the sets of  real numbers, $n$-dimensional vectors, and $n$-by-$m$-dimensional matrices, respectively. The set of natural numbers is represented by $\mathbb{N}$. $\mathbb{E}(\cdot)$ and $\mathbb{P}(\cdot)$ denote the expectation of a random variable and the probability of an event, respectively. For a matrix $A$, $A^T$ is its transpose, and $\|A\|$ refers to its spectral norm. Positive semidefinite (definite) matrices are denoted by $A\succeq0$ ($A\succ0$). The notations $I$ and $\bm 0$ represent the identity and zero matrices with compatible dimensions, respectively. When there is no confusion, $\times$ represents the direct product. Finally, $f^{(i)}(\cdot)$ denotes the $i$-fold composition of a function, defined recursively as $f^{(i)}(\cdot) = f^{(i-1)}[f(\cdot)]$.

\section{Problem Setup}
Consider $N$ discrete linear time-invariant (LTI) processes with $N$ sensors (Fig.~\ref{fig:framework}). The $i$-th system is observed by the $i$-th sensor, i.e.,
\begin{equation*}
x^{(i)}_{k+1} 	= A_ix^{(i)}_k + w^{(i)}_k,\quad y^{(i)}_k		=C_ix^{(i)}_k + v^{(i)}_k,
\end{equation*}
where $k$ is the discrete time index, $x_k^{(i)}\in\mathbb{R}^{n_i}$ is the system state,  $y^{(i)}_k\in\mathbb{R}^{m_i}$ is the sensor measurement vector, $w_k^{(i)}$ and $v_k^{(i)}$ are mutually {uncorrelated} Gaussian white noise with covariances $W_i\succeq 0$ and $V_i\succ0$, respectively. The initial state $x_0^{(i)}$ is assumed to be zero-mean Gaussian with covariance $P_0^{(i)}\succeq 0$, and uncorrelated with $w_k^{(i)}$ and $v_k^{(i)}$ for all $k$. We assume $(A_i,C_i)$ is detectable and $(A_i,\sqrt{W_i})$ is stabilizable. Additionally, we assume $\|A_i\|\geq 1$ (This is discussed later in Remarks~\ref{rmk:assumption1} and~\ref{rmk:stable}).\footnote{The same assumption has been adopted in~\cite{peng2017optimal,zhang2020false}.}

\begin{figure}[!tbp]
	\centering
	\includegraphics[width=0.9\linewidth]{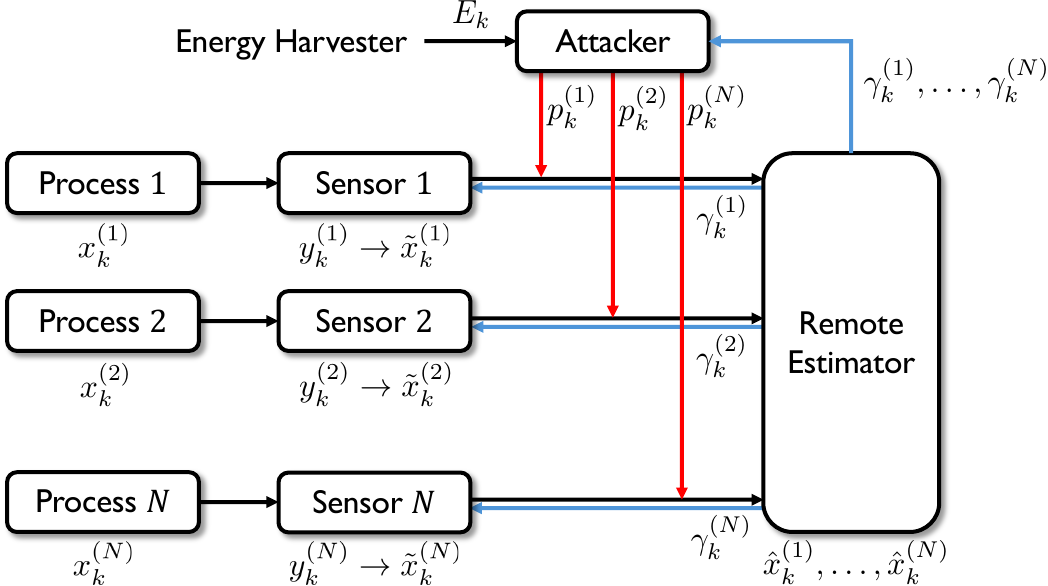}
	\caption{System diagram.}
	\label{fig:framework}
\end{figure}
\subsection{Sensor Local State Estimation}
As shown in Fig.~\ref{fig:framework}, each sensor is assumed to be ``smart'' and capable of running a Kalman filter locally to estimate the system state based on its accumulated measurements, i.e.,
\begin{align*}
\tilde{x}_k^{(i)} 			&\triangleq \mathbb{E}\left[x_k^{(i)}|y_0^{(i)},\dots,y_k^{(i)}\right],\\
\tilde{P}_k^{(i)} 			&\triangleq \mathbb{E}\left[(x_k^{(i)}-\tilde{x}_k^{(i)})(x_k^{(i)}-\tilde{x}_k^{(i)})^T|y_0^{(i)},\dots,y_k^{(i)}\right],
\end{align*}
where $\tilde{x}_k^{(i)}$ and $\tilde{P}_k^{(i)}$ are the local estimate and the corresponding estimation error covariance, respectively.

For notational simplicity, define the following functions:
\begin{align*}
\tilde{g}_i(X)	&\triangleq X-XC_i^T[C_iXC_i^T+V_i]^{-1}C_iX,\\
h_i(X)		&\triangleq A_iXA_i^T+W_i,\qquad g_i(X)		\triangleq \tilde{g}_i\circ h_i(X).
\end{align*}
It is well-established that $\tilde{P}_k^{(i)}$ converges exponentially to a steady-state $\bar{P}_{i}$, which can be uniquely computed by solving the equation $g_i(X)=X$~\cite{anderson2012optimal}. As we herein focus on the asymptotic performance over an infinite time horizon, we assume without loss of generality that the Kalman filter at the sensors have reached the steady state at $k=0$, i.e., $\tilde{P}_k^{(i)}=\bar{P}_{i}$.

\subsection{Remote State Estimation}
After obtaining the local estimate, the sensors transmit their data to a remote state estimator via wireless communication channels. Due to the unreliability of wireless communications, the packet transmission suffers from dropouts. Let $\gamma_k^{(i)}\in\{0,1\}$ denote the binary variable indicating whether the remote estimator receives the sensor's local estimate at time $k$, i.e., $\gamma_k^{(i)}=1$ if the estimator receives $\tilde{x}_k^{(i)}$, and $\gamma_k^{(i)}=0$ otherwise.

Define $\tau_k^{(i)}$ as the duration between $k$ and the most recent reception of the sensor's estimate, i.e.,
\begin{equation}\label{eq:tau}
\tau_k^{(i)}\triangleq\min\{t\geq 0:\gamma^{(i)}_{k-t}=1\}.
\end{equation}
Since we are interested in the asymptotic behavior of the system, we assume $\gamma_0^{(i)}=1$ without loss of generality.

Given $\gamma_k^{(i)}$ and $\tau_k^{(i)}$, the minimum mean square error (MMSE) estimate and the relative error covariance at the remote estimator obey the stochastic recursions~\cite{shi2010kalman}: $\hat{x}^{(i)}_k=\tilde{x}^{(i)}_k$ and $P^{(i)}_k=\bar{P}_{i}$ if $\gamma^{(i)}_k=1$; $\hat{x}^{(i)}_k=A_i\hat{x}_{k-1}^{(i)}$ and $P^{(i)}_k=h_i(P^{(i)}_{k-1})=h_i^{\tau_k^{(i)}}(\bar{P}_{i})$ otherwise.

\subsection{Self-Sustainable Attack}
Equipped with energy harvesters, the attacker can launch jamming attacks, thus reducing the estimation accuracy of the remote estimator by utilizing the harvested energy. 

Let $b_k$ represent the battery level of the attacker at time $k$, and let $b_{\max}$ denote the maximum battery capacity. At each time $k$, the attacker determines the amount of jamming attack power $p_k^{(i)}\geq 0$ to inject into the $i$-th sensor-estimator communication link. The power $p_k^{(i)}$ is determined to maximize the estimation error covariance at the remote estimator, which will be introduced later.

\begin{figure}[!tbp]
	\centering
	\includegraphics[width=0.9\linewidth]{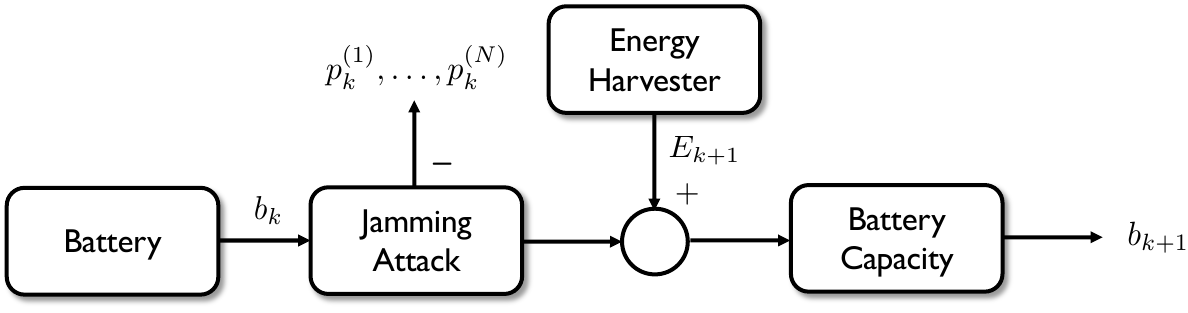}
	\caption{Battery dynamics of the attacker.}
	\label{fig:battery}
\end{figure}

Since the jamming power is limited by the available energy in the battery and the attacker's transmission capacity, we have the following constraints on $p_k^{(i)}$:
\begin{equation}\label{eq:p_con}
p_k^{(i)}\leq p_{\max},\qquad \sum_{i=1}^N p_k^{(i)}\leq b_k,
\end{equation}
where $p_{\max}$ represents the maximum power the attacker can allocate to each channel.

On the other hand, let $E_k$ denote the amount of energy harvested by the attacker during the interval between time $k$ and $k+1$. The battery dynamics (Fig.~\ref{fig:battery}) are given by~\cite{ho2012optimal,leong2018optimal}
\begin{equation}\label{eq:bat_model}
b_{k+1} = \min\{b_k-\sum_{i=1}^N p_k^{(i)}+\lfloor E_{k}\rfloor,b_{\max}\},
\end{equation}
where $\lfloor\cdot\rfloor$ is the floor function that returns the greatest integer less than or equal to the given input.

Generally speaking, $E_k$ is time-varying and may exhibit temporal correlation. For instance, the amount of solar energy can vary significantly depending on the time of day or weather conditions. In this paper, we assume that $E_k$ is represented by a finite-state Markov model~\cite{ho2012optimal}, where the set of possible states is denoted by $\mathcal{E}=\{e_1,\dots,e_\ell\}$. The transition probability is
\begin{equation}\label{eq:e_trans}
\mathbb{P}(E_{k+1}=e_q|E_{k}=e_{p})=t^e_{e_p,e_q}.
\end{equation}

{
While models like renewal processes~\cite{li2016power}, fluid models~\cite{gautam2015efficiently}, or Brownian motion~\cite{abdelrahman2016diffusion} exist, they often either fail to capture the temporal dependence of solar radiation or introduce unnecessary complexity. Empirical studies have demonstrated that finite-state Markov models provide accurate approximations for solar energy harvesting~\cite{ho2012optimal}. Consequently, they have been widely adopted in this domain. 
}

\subsection{Wireless Communication}
Let $\lambda_k^{(i)}\triangleq\mathbb{P}(\gamma_k^{(i)}=1)$ represent the packet arrival rate for the wireless transmission. It is influenced by two factors: i) the attackers' jamming power $p_k^{(i)}$; ii) the channel gains of both the sensor-estimator and attacker-estimator links.

Denote $H^{(i)}_k$ and $G_k^{(i)}$ as the channel gains of the sensor-estimator and attacker-estimator links at time $k$, respectively.  Assume the gains are Markovian with states $\mathcal{S}=\{s_1,\dots,s_l\}$. The transition probabilities are:
\begin{equation}\label{eq:transition}
\begin{aligned}
\mathbb{P}	&(H^{(i)}_{k+1}=s_q|H^{(i)}_{k}=s_{p})\\
			&\qquad\qquad=\mathbb{P}(G^{(i)}_{k+1}=s_q|G^{(i)}_{k}=s_{p})=t^s_{s_p,s_q}.
\end{aligned}
\end{equation}

\begin{rem}{\rm
The finite-state Markov model can be regarded as an extension of the well-known Gilbert-Elliott model with higher model accuracy~\cite{wang1995finite}.
}\end{rem}
\begin{rem}{\rm
While the transition probability is not necessary to be identical across systems and communication links as shown in~\eqref{eq:transition}, for better presentation of the results, we assume it is identical without loss of generality.
}\end{rem}
Consider a unit time slot. Under the jamming attack, the SINR for the $i$-th sensor at time $k$ is given by
\begin{equation}\label{eq:sinr}
{\rm SINR}_k^{(i)}\triangleq\frac{H_k^{(i)}}{G_k^{(i)}p_k^{(i)}+\sigma_i^2},~  p_k^{(i)}\leq p_{\max}, ~ \sum_{i=1}^N p_k^{(i)}\leq b_k,
\end{equation}
where $\sigma_i^2$ represents the power of additive white Gaussian noise in the $i$-th channel.
{
\begin{rem}{\rm
In this model~\eqref{eq:sinr}, the sensor's transmission power is assumed to be constant and normalized to unity. This allows us to focus on the optimal power allocation of the attacker. Extending the model to include adaptive sensor power transmission would lead to a Markov game formulation~\cite{littman1994markov}. We defer it as our future research.
}\end{rem}}

Using error-detection mechanisms such as cyclic redundancy checks, the packet arrival rate is determined by the associated SINR~\cite{alma991002544119703412}, i.e., $\lambda_k^{(i)}=f({\rm SINR}_k^{(i)})$, where $f(\cdot):\mathbb{R}\mapsto[0,1]$ is a non-decreasing function specified by the modulation mode used by the sensors and the estimator. {While our results hold for a general non-decreasing function $f(\cdot)$, we present the following example for illustrative purposes.}
\begin{exmp}\label{exam:qam}{\rm
Considering a quadrature amplitude modulation (QAM) scheme, we have $f({\rm SINR}_k^{(i)})=[1-\alpha Q(\sqrt{b {\rm SINR}_k^{(i)}})]^B$, where $\alpha$, $b$ and $B$ are channel parameters and $Q(x)\triangleq\frac{1}{\sqrt{2\pi}}\int_{x}^{\infty}\exp(-\delta^2/2)d\delta$.
}\end{exmp}
{
\begin{rem}\label{rmk:assumption1}{\rm
It is worth mentioning that when $\|A_i\|<1$, the estimation error converges to a bounded value even under persistent worst-case attacks, i.e., with $p_k^{(i)}= p_{\max}$ for all $k$ and $i$. In other words, the attacks cannot destabilize the system. To ensure the problem is well-posed, we assume $\|A_i\|\geq1$ herein and defer the study of the case $\|A_i\|<1$ to our future work. Nevertheless, simulation results (see Section 6) indicate that the proposed algorithms remain applicable even when $\|A_i\|<1$.
}\end{rem}}
\subsection{Problem of Interest}
Assume the attacker knows its battery level $b_k$, the harvested energy $E_k$ and its transition probability~\eqref{eq:e_trans}. Additionally, the attacker knows the system dynamics, i.e., $A_i$, $C_i$, {$W_i$, and $V_i$}, and can obtain $\tau^{(i)}_{k-1}$ by eavesdropping on the feedback channels from the estimator to the sensors.\footnote{The same assumption has been adopted in prior studies on related topics~\cite{peng2017optimal,ren2018attack}.} Regarding the channel information, two scenarios are considered:\\
(1) \underline{\it Perfect channel information}: the attacker has access to the values of $H^{(i)}_k$, and $G^{(i)}_k$ via channel estimation and knows the transition probability~\eqref{eq:transition} perfectly.\\
(2) \underline{\it Unknow channel model}: the attacker knows the values of $H^{(i)}_k$, and $G^{(i)}_k$ but not the transition probability.\footnote{This assumption is also commonly found in numerous works within the literature~\cite{lyu2020irs,zhang2017wireless}.}
\begin{rem}{\rm
	The above assumptions imply a strong attacker. This follows Kerckhoffs's principle, which states that security should not rely on secrecy~\cite{van2014encyclopedia}. Studying the ``worst-case" attack helps us understand the system's security level. If instead the attacker has limited information, the MDP model (Section~\ref{chap:results}) should be extended to a partially observable MDP (POMDP) model~\cite{krishnamurthy2016partially}. 
}\end{rem}
{
\begin{rem}{\rm
The ``unknown channel model'' scenario is motivated by the fact that channel estimation can be performed on a fast timescale using pilot signals~\cite{liu2014channel}, whereas learning the statistical transition model typically requires much longer observation periods. The scenario where both factors are unknown can be modeled as a POMDP as well. While POMDPs offer a more general framework, they significantly increase computational complexity due to the need for belief-state tracking~\cite{krishnamurthy2016partially}. We therefore leave it for future research.
}\end{rem}}

The attacker aims to smartly utilize the harvested energy to disturb the estimation process of the remote estimator. More specifically, the objective function is
\begin{equation}\label{eq:cost}
J(p)=\liminf_{T\to\infty}\frac{1}{T+1}\sum_{k=0}^{T}\sum_{i=1}^{N}\tr\left\{\mathbb{E}\left[P_k^{(i)}\right]\right\},
\end{equation}
where $p=\{p_0^{(1)},\dots,p_0^{(N)},\dots,p_k^{(1)},\dots,p_k^{(N)},\dots\}$ is the jamming power subject to the energy constraint~\eqref{eq:p_con}.

The problem addressed in this paper can be summarized as an optimization problem as follows:
\begin{align}
\max_{p} 	&\quad J(p)\label{eq:problem}\\
{\rm s.t.}		&\quad 0\leq p_k^{(i)}\leq p_{\max},~ i = 1,\dots,N,~\forall k\in\mathbb{N},\notag\\
			&\quad \sum_i p_k^{(i)}\leq b_k,~\forall k\in\mathbb{N},\notag\\
			&\quad b_{k+1} = \min\{b_k-\sum_{i} p_k^{(i)}+\lfloor E_{k}\rfloor,b_{\max}\}, ~\forall k\in\mathbb{N}.\notag
\end{align}

\section{MDP Formulation and Structural Results}\label{chap:results}
In this section, we solve the problem enunciated in the preceding section with an MDP model and prove the existence of an optimal stationary policy under mild assumptions. Additionally, we establish some structural properties of the optimal policy, which enables efficient computation in both scenarios of perfect and unknown channel information.
\subsection{MDP Design}
Since in practice, the transmitter is typically programmed to operate with a finite number of transmission power levels only~\cite{fu2012practical}, we herein assume the attacker chooses $p_k^{(i)}$ discretely, i.e., $p_k^{(i)}\in\mathcal{P}=\{0,1,\dots, \lfloor p_{\max}\rfloor\}$.

The MDP formulation requires the following definitions.

\noindent 1) \emph{State:} The attacker's state encapsulates its information for decision-making. Since the power allocation strategy depends on the channel state $H_k^{(i)}$, $G_k^{(i)}$, the estimation error covariance of the remote estimator $P_k^{(i)}$, the battery level $b_k$ and the harvested energy $E_k$, the state at the beginning of time $k$, i.e., $\phi_k$, is defined as follows: $\phi_k=(b_k, E_k, \Lambda_k^{(1)},\dots,\Lambda_k^{(N)})$, where $\Lambda_k^{(i)}=(H_k^{(i)},G_k^{(i)},\tau_{k-1}^{(i)})$.

We note that the state space $\mathbb{S}$ is countably infinite, i.e., $\mathbb{S}=\mathcal{B}\times\mathcal{E}\times\prod_{i=1}^N\left\{H^{(i)},G^{(i)},\tau^{(i)})\right\}$, where $\mathcal{B}=\{0,1,\dots,b_{\max}\}$, $H^{(i)},G^{(i)}\in\mathcal{S}$ and $\tau^{(i)} \in\mathbb{N}$.

\noindent 2) \emph{Action:} The action $a_k$ is the jamming power at time $k$, i.e., $a_k=[a_k^{(1)},\dots,a_k^{(N)}]$, where $a_k^{(i)}=p_k^{(i)}\in\mathcal{P}$.

The available action for state $\phi_k$ is $\mathbb{A}_{\phi_k}=\{(a^{(1)}_k,\dots,\allowbreak a^{(N)}_k)| \sum_{i}a^{(i)}_k\leq b_k,~a^{(i)}\in\mathcal{P}, ~i=1,\dots,N\}$. Then the action space is
$\mathbb{A}=\cup_{k=1}^\infty\mathbb{A}_{\phi_k}=
\{(a^{(1)},\dots,a^{(N)})|	\sum_{i}a^{(i)}\allowbreak\leq b_{\max},~a^{(i)}\in\mathcal{P}, ~i=1,\dots,N\}$.

{
Note that while the available action set $\mathbb{A}_{\phi_k}$ is constrained by the instantaneous battery level $b_k$, the global action space $\mathbb{A}$ is defined relative to the maximum battery capacity $b_{\max}$. This ensures that $\mathbb{A}$ remains a state-independent set that encompasses all possible actions $a_k$ regardless of the current state.}
\begin{rem}{\rm
As the state $\phi_k$ is defined at the start of time $k$, $\tau_{k}^{(i)}$ is unavailable to the attacker for decision-making.
}\end{rem}
\noindent 3) \emph{Transition Probability:} {Define $\mathbb{T}:\mathbb{S}\times\mathbb{A}\times\mathbb{S}\mapsto[0,1]$ as the transition function. Let $\phi=(b, E, \Lambda^{(1)},\dots,\Lambda^{(N)})$, where $\Lambda^{(i)}=(H^{(i)},G^{(i)},\tau^{(i)})$. The transition probabilities can be written as
\begin{equation*}
\begin{aligned}
&\mathbb{P}(\phi_{k+1}=\phi'|\phi_k=\phi,a_k=a)\\
&\quad=\mathbb{P}(b_{k+1}=b',E_{k+1}=e'|b_k=b,E_k=e,a_k=a)\\
&\quad\quad\times\prod_{i=1}^N\mathbb{P}(\Lambda^{(i)}_{k+1}=\Lambda^{(i)'}|\Lambda^{(i)}_k=\Lambda^{(i)},a^{(i)}_k=a^{(i)}),	
\end{aligned}
\end{equation*}
where $\mathbb{P}(b_{k+1}=b',E_{k+1}=e'|b_k=b,E_k=e,a_k=a)$ characterizes the battery dynamics, and $\mathbb{P}(\Lambda^{(i)}_{k+1}=\Lambda^{(i)'}|\Lambda^{(i)}_k=\Lambda^{(i)},a^{(i)}_k=a^{(i)})$ captures the channel dynamics and the packet reception process. Using the model~\eqref{eq:bat_model}-\eqref{eq:transition}, these probabilities are given by
\begin{equation*}
\begin{aligned}
&\mathbb{P}(b_{k+1}=b',E_{k+1}=e'|b_k=b,E_k=e,a_k=a)\\
=&\begin{cases}
t^e_{e,e'}, 	& \text{if~\eqref{eq:bat_model} holds};\\
0, 				&\text{otherwise},
\end{cases}\\
	&\mathbb{P}(\Lambda^{(i)}_{k+1}=\Lambda^{(i)'}|\Lambda^{(i)}_k=\Lambda^{(i)},a^{(i)}_k=a^{(i)})\\
=	&t^s_{H^{(i)},H^{(i)'}}t^s_{G^{(i)},G^{(i)'}}
\begin{cases}
	\lambda^{(i)},  		&{\rm if~}\tau^{(i)'}=0;\\
	1-\lambda^{(i)},		&{\rm if~}\tau^{(i)'}=\tau^{(i)}+1;\\
	0, 					&{\rm otherwise}.
\end{cases}
\end{aligned}
\end{equation*}
More specifically, given $e'$, $b'$ is uniquely determinted by $(b,a)$ according to~\eqref{eq:bat_model}. Hence, the transition is feasible only if~\eqref{eq:bat_model} is satisfied, in which case its probability is governed by the energy arrival process, i.e., $t^e_{e,e'}$. When ~\eqref{eq:bat_model} is not satisfied, the probability is zero. On the other hand, according to~\eqref{eq:tau}, $\tau^{(i)'}$ evolves depending on the packet reception indicator $\gamma^{(i)}$: it resets to $0$ if a packet is successfully received, and increases to $\tau^{(i)'}=\tau^{(i)}+1$ if not received. All the other values of $\tau^{(i)'}$ are impossible and thus have zero probability.}

\noindent 4) \emph{One-Stage Reward Function:} Considering the reward function defined in~\eqref{eq:cost}, the one-stage reward function for state $\phi_k$ and action $a_k$ becomes
\begin{equation}\label{eq:reward}
\begin{aligned}
&r(\phi_k,a_k)	=\sum_{i=1}^N \tr\left\{\mathbb{E}\left[P_k^{(i)}|\phi_k,a_k\right]\right\}\\
			=&\sum_{i=1}^N\left\{\lambda_k^{(i)}\tr[\bar{P}_{i}]+(1-\lambda_k^{(i)})\tr[h_i^{\tau^{(i)}_{k-1}+1}(\bar{P}_{i})]\right\}.
\end{aligned}
\end{equation}
The objective function in~\eqref{eq:cost} is then given by the average expected sum of one-stage rewards, i.e., $J(p)=\liminf_{T\to\infty}\frac{1}{T+1}\sum_{k=0}^{T}r(\phi_k,a_k)$.

{Since both the transition probability and the reward function are time-homogeneous, we omit the time index $k$ for notational simplicity. We also shorthand the transition probability $\mathbb{P}(\phi_{k+1}=\phi'|\phi_k=\phi,a_k=a)$ as $\mathbb{P}(\phi'|\phi,a)$. Additionally,  we denote current values with $(\cdot)$ and their values at the next time step with $(\cdot)'$.}

Denote $J^*$ as the optimal value of the objective, i.e., $J^*=\sup_{p}J(p)$. Then $J^*$ satisfies the following average-cost Bellman equation:
\begin{equation}\label{eq:bellman}
J^*+V(\phi)=\max_{a\in\mathbb{A}_\phi}\{r(\phi,a)+\sum_{\phi'\in\mathbb{S}}\mathbb{P}(\phi'|\phi,a)V(\phi')\},
\end{equation}
where $V(\cdot):\mathbb{S}\mapsto\mathbb{R}$ is the relative value function.

However, the Bellman equation~\eqref{eq:bellman} may fail to have a solution~\cite{bertsekas2011dynamic}, for instance, when  $J^*$ is infinite. Let $\theta(\cdot)$ be a deterministic policy that maps states to actions, i.e., $\mathbb{S}\mapsto\mathbb{A}$, and let $\Theta$ denote the set of all feasible policies. In the following, we will establish the existence of an optimal deterministic stationary policy that solves~\eqref{eq:bellman}.
\begin{assum}\label{assumption:stable}{\rm
For every sensor $i=1,\dots,N$, assume there exisits $\kappa_i\in[0,1)$ such that
\begin{equation*}
\inf_{\Xi^{(i)}}\sum_{\Xi^{(i)'}}\left[f({\rm SINR}^{(i)'}_{\min})t_{H^{(i)},H^{(i)'}}t_{G^{(i)},G^{(i)'}}\right]\geq1-\frac{\kappa_i}{\|A_i\|^2},
\end{equation*}
where the quantities are defined in~\eqref{eq:transition}-\eqref{eq:sinr} and as follows:
\begin{equation*}
\Xi^{(i)}=(H^{(i)},G^{(i)}),\quad {\rm SINR}^{(i)}_{\min}={H^{(i)}}/[\lfloor p_{\max} \rfloor G^{(i)}+\sigma_i^2].
\end{equation*}
}\end{assum}
\begin{lem}\label{lemma:bound}{\rm
If Assumption~\ref{assumption:stable} holds, then $\mathbb{E}[P_k^{(i)}]$ is exponentially bounded for any attack allocations, i.e.,
\begin{equation*}
\mathbb{E}[\tr(P_k^{(i)})]\leq c_i\kappa_i^k+d_i, ~i=1,\dots,N, ~k\in\mathbb{N},
\end{equation*}
for some finite constants $c_i$ and $d_i$.
}\end{lem}
\begin{pf}
The proof follows directly from~\cite[Theorem 1]{quevedo2012state}, thus omitted here.
\end{pf}
\begin{rem}{\rm
Intuitively, Assumption~\ref{assumption:stable} ensures that, starting from any channel condition $\Xi_{k-1}^{(i)}$, the expected packet arrival rate does not become excessively low, even under worst-case attacks. {This guarantees that $\mathbb{E}[P_k^{(i)}]$ remains bounded as $k\to\infty$, thereby ruling out trivial scenarios in which the attacker could persistently jam a single channel and drive the estimation error covariance to infinity. Moreover, it guarantees a finite solution $J^*$ to~\eqref{eq:bellman}. This assumption is standard in remote estimation literature and are typically satisfied in practical systems with moderate communication reliability.}
}\end{rem}
\begin{thm}\label{thm:stationary}{\rm
Assume Assumption~\ref{assumption:stable} holds. There exists a deterministic stationary policy $\theta^*$, a constant $J^*$, and a function $V(\phi)$ which solve the Bellman equation~\eqref{eq:bellman}. The optimal policy $\alpha^*(\phi)=\theta^*(\phi)$ is determined by
\begin{equation*}
\alpha^*(\phi)=\arg\max_{\alpha\in\mathbb{A}_\phi}\{r(\phi,\alpha)+\sum_{\phi'\in\mathbb{S}}\mathbb{P}(\phi'|\phi,\alpha)V(\phi')\}.
\end{equation*}
}\end{thm}
\begin{pf}
See Appendix~\ref{apx:stationary}.
\end{pf}
\begin{rem}\label{rmk:stable}{\rm
Constructing a measurable function $\Psi(\cdot)$ that simultaneously satisfies the intricate conditions (I) and (II) in Appendix~\ref{apx:stationary} is fundamentally challenging. In particular, when $\|A_i\|<1$, the condition $\Psi(\cdot)\geq1$ no longer holds for the constructed measurable function.
}\end{rem}

\subsection{Structural Results}\label{sec:structure}
In this part, we demonstrate that the optimal policy exhibits monotonic structures, which allows for a significant reduction in the computational complexity of the algorithms introduced in subsequent sections.

Following ideas in~\cite[Section 8.11]{puterman2014markov}, we denote the $Q$-factor of a state-action pair as
\begin{equation}\label{eq:q factor}
Q(\phi,a)\triangleq r(\phi,a)+\sum_{\phi'\in\mathbb{S}}\mathbb{P}(\phi'|\phi,a)V(\phi')-J^*.
\end{equation}
Combining $V(\phi)=\max_{a\in\mathbb{A}}Q(\phi,a)$ with~\eqref{eq:bellman}, we obtain
\begin{equation}\label{eq:q update}
J^*+Q(\phi,a)=r(\phi,a)+\sum_{\phi'\in\mathbb{S}}\mathbb{P}(\phi'|\phi,a)\max_{a'\in\mathbb{A}_{\phi'}}Q(\phi',a').
\end{equation}

To introduce ordering, with a slight abuse of notation, we denote $\phi_+=(b,E,\Lambda_+^{(1)},\dots,\Lambda_+^{(N)})$ and $\phi_-=(b,E,\Lambda_-^{(1)},\dots,\Lambda_-^{(N)})$ where each component $\Lambda_+^{(i)}$ and $\Lambda_-^{(i)}$ is given by $\Lambda^{(i)}_+=(H^{(i)},G^{(i)},\tau^{(i)}_+)$ and $ \Lambda^{(i)}_-=(H^{(i)},G^{(i)},\tau^{(i)}_-)$. Here $b$, $E$, $H^{(i)}$ and $G^{(i)}$ represent arbitrary but fixed elements. We say $\phi_-\preceq_{\tau}\phi_+$ if $\tau_-^{(i)}\leq\tau_+^{(i)}$ for all $i$. Note that $\preceq_{\tau}$ is a partial order defined on $\mathbb{S}$.

Now we are ready to present the monotonicity properties of $V(\phi)$ and $Q(\phi,a)$, respectively.

\begin{lem}[Monotonicity]\label{lemma:monototic}{\rm
The function $V(\phi)$ is monotonically increasing in $\tau^{(i)}$ for all $i$. In other words, for all pairs $(\phi_-,\phi_+)$ such that $\phi_-\preceq_{\tau}\phi_+$, we have $V(\phi_-)\leq V(\phi_+)$.
}\end{lem}
\begin{pf}
See Appendix~\ref{apx:monototic}.
\end{pf}

\begin{lem}[Monotonicity]\label{lemma:monototic q}{\rm
The $Q$-factor $Q(\phi,a)$ is monotonically increasing in $\tau^{(i)}$ for all $i$. In other words, for all pairs $(\phi_-,\phi_+)$ such that $\phi_-\preceq_{\tau}\phi_+$ and $\forall a\in\mathbb{A}_{\phi_+}\cap\mathbb{A}_{\phi_-}$, we have 
\begin{equation}\label{eq:monototic q}
Q (\phi_-,a)\leq Q (\phi_+,a).
\end{equation}
}\end{lem}
\begin{pf}
The proof follows directly from Lemma~\ref{lemma:monototic} and~\eqref{eq:q factor}, thus omitted here.
\end{pf}
Fixing the other components of the state $\phi$ except $\tau^{(i)}$, with a bit of abuse of notations, we denote the $Q$-factor in~\eqref{eq:q factor} as $Q(\tau^{(i)},a)$. Additionally, denote $a_{-,i}=(a_{-,i}^{(1)},\dots,a_{-,i}^{(N)})$ and $a_{+,i}=(a_{+,i}^{(1)},\dots,a_{+,i}^{(N)})$, where the inequality $a_{-,i}^{(i)}\leq a_{+,i}^{(i)}$ holds for the specific index $i$. Then we are prepared to present the superadditivity property of $Q(\tau,a)$.
\begin{lem}[Superadditivity]\label{lemma:super}{\rm
The $Q$-factor $Q(\tau^{(i)}, a)$ is superadditive in $(\tau^{(i)}, a^{(i)})$. In other words, for all $\tau^{(i)}_-\leq  \tau^{(i)}_+$ and for all $ a_{-,i},a_{+,i}$ satisfying $a_{-,i}^{(i)}\leq a_{+,i}^{(i)}$, the following inequality holds:
\begin{equation}\label{eq:super}
 Q (\tau^{(i)}_+,a_{+,i})- Q (\tau^{(i)}_+,a_{-,i})\geq Q (\tau^{(i)}_-,a_{+,i})- Q (\tau^{(i)}_-,a_{-,i}).
\end{equation}
}\end{lem}
\begin{pf}
See Appendix~\ref{apx:additive}.
\end{pf}

{
Unlike monotonicity, which extends globally due to the additive reward structure in~\eqref{eq:reward}, superadditivity is established only locally for each sensor $i=1,\dots,N$. This is because the shared battery constraint $\sum_{i=1}^N a_k^{(i)}\leq b_k$ couples the action variables across sensors, thereby generally preventing global superadditivity from holding. Nevertheless, the local result is sufficient for proving the threshold property in the following Theorem~\ref{thm:structure}.
}
\begin{thm}\label{thm:structure}{\rm
The optimal allocated power for the $i$-th sensor, i.e., $\alpha^{(i)*}(\phi)$, is non-decreasing in $\tau^{(i)}$, where $i=1,\dots,N$. In other words, $\forall\phi_-\preceq_{\tau}\phi_+$, we have $\alpha^{(i)*}(\phi_-)\leq\alpha^{(i)*}(\phi_+)$.
}\end{thm}
\begin{pf}
See Appendix~\ref{apx:structure}.
\end{pf}
The structural results developed above can be leveraged to enhance the efficiency of our algorithms. In Section~\ref{sec:perfect} and Section~\ref{sec:unknown}, we will provide a detailed discussion on how to utilize these findings.

\section{Perfect Channel Information}\label{sec:perfect}
By virtue of Theorem~\ref{thm:stationary}, if the attacker has perfect channel information, we can solve~\eqref{eq:bellman} by the relative value iteration algorithm~\cite{bertsekas2011dynamic}. The iterative recursion is
\begin{equation}\label{eq:iteration}
\begin{aligned}
D_t(\phi)	&=\max_{a\in\mathbb{A}_\phi}\{r(\phi,a)+\sum_{\phi'\in\mathbb{S}}\mathbb{P}(\phi'|\phi,a)D_{t-1}(\phi')\}\\
		&-\max_{a\in\mathbb{A}_{\phi_f}}\{r(\phi_f,a)+\sum_{\phi'\in\mathbb{S}}\mathbb{P}(\phi'|\phi_f,a)D_{t-1}(\phi')\},
\end{aligned}
\end{equation}
where $t$ is the iteration index, $\phi_f\in\mathbb{S}$ is an arbitrarily chosen reference state, {and $D_{0}(\cdot)$ is initialized arbitrarily subject to the normalization condition $D_0(\phi_f)=0$.}
\begin{rem}{\rm{
For average-cost problems, the commonly used value iteration algorithm, i.e., $D_t(\phi)	=\max_{a\in\mathbb{A}_\phi}\{r(\phi,a)+\sum_{\phi'\in\mathbb{S}}\mathbb{P}(\phi'|\phi,a)D_{t-1}(\phi')\}$, may diverge to infinity~\cite{bertsekas2011dynamic}. The modified relative value iteration in~\eqref{eq:iteration} avoids this issue by normalizing the value function with respect to the reference state. In particular, it enforces $D_t(\phi_f)= 0$ at every iteration.}
}\end{rem}
{
\floatname{algorithm}{Algorithm}
\begin{algorithm}[!htbp]

\caption{Relative Value Iteration for Perfect Channel Information}
\label{alg:rvi}
\begin{algorithmic}[1]
\State Initialize $D_0(\phi)$ arbitrarily for all $\phi\in\mathbb{S}$, subject to $D_0(\phi_f)=0$
\State Set $t\gets 1$
\Repeat
    \For{each state $\phi\in\mathbb{S}$}
        \State Compute $D_t(\phi)$ from $D_{t-1}(\phi)$ via~\eqref{eq:iteration}
    \EndFor
    \State $t\gets t+1$
\Until{convergence}
\State \Return $V(\phi)\gets D_t(\phi)$
\end{algorithmic}
\end{algorithm}}

\begin{prop}\label{prop}{\rm
As $t\to\infty$, $D_t(\phi_f)$ and $D_t(\phi)$ in~\eqref{eq:iteration} converge to the optimal solution $J^*$ and $V(\phi)$ of the Bellman equation~\eqref{eq:bellman}, respectively.
}\end{prop}
\begin{pf}
By utilizing~\cite[Proposition 4.3.2]{bertsekas2011dynamic} and selecting a state $z=(b,E,\Lambda^{(1)},\dots,\Lambda^{(N)})$, where $\Lambda^{(i)}=(H^{(i)},G^{(i)},0)$, the proof is straightforward and is thus omitted here.
\end{pf}

Proposition~\ref{prop} indicates relative value iteration yields the optimal policy. {Algorithm~\ref{alg:rvi} summarizes the detailed procedure for computing $D_t(\phi)$ and thus the optimal policy.\footnote{Algorithm~\ref{alg:rvi} has a theoretical per-iteration complexity $\mathcal{O}(|\mathbb{S}|^2 |\mathbb{A}|)$, where $|\cdot|$ is the cardinality of the set. Fortunately, the transitions for $\tau^{(i)}$ is sparse, i.e., either increases by $1$ or resets to $0$. Since the state space of $\tau^{(i)}$ is typically much larger than that of the channel states, the complexity is approximately reduced to $\mathcal{O}(|\mathbb{S}| |\mathbb{A}|)$.}} However, this procedure can be computationally expensive when the state space is large. The monotonic property in Theorem~\ref{thm:structure} indicates that the optimal power allocations switch between certain thresholds. While these thresholds lack closed forms, gradient-estimate-based methods~\cite{nourian2014optimal} can be employed to search for them. This significantly enhances the computational efficiency of determining the optimal policy~\cite{puterman2014markov}.

Additionally, it is important to note that the countably infinite nature of the state space makes it impractical to solve the Bellman equation directly. In practice,  $\tau^{(i)}$ is truncated at some $L$, with all larger values mapped to $L$. Given the monotonic nature of the optimal policy, the solution to the truncated problem coincides with the true optimal solution when $L$ is sufficiently large.

\section{Unknown Channel Model}\label{sec:unknown}
When the transition probabilities are unknown to the attacker, solving the Bellman equation via relative value iteration becomes inapplicable. However, as an alternative approach, we can employ a stochastic approximation algorithm that does not require prior knowledge of the transition probabilities. In the remainder of this section, we first outline the standard stochastic approximation with standard updates. Building on this foundation, we propose a novel structural update mechanism that accelerates convergence by leveraging the structural results developed in Section~\ref{sec:structure}.

\subsection{Standard Update}\label{sec:std_update}
At each time $k$ and state $\phi$, the action $a$ is {selected} according to an $\epsilon$-greedy strategy based on the $k$-th iteration of the $Q$-factor\footnote{Herein the time index $k$ and the algorithm iteration number $t$ can be regarded as identical.}, i.e.,
\begin{equation}\label{eq:e_greedy}
a=
\begin{cases}
\arg\max_{a'\in\mathbb{A}_\phi}Q_k(\phi,a'),			&\text{with probability $1-\epsilon$};\\
\text{any action in $\mathbb{A}_\phi$},	&\text{with probability $\epsilon$}.
\end{cases}
\end{equation}
{
After executing action $a$, we observe the next state $\phi'$. Since the transition probability $\mathbb{P}(\phi'|\phi,a)$ is unknown, the expectation term in~\eqref{eq:q update}, i.e., $\sum_{\tilde{\phi}\in\mathbb{S}}\mathbb{P}(\tilde{\phi}|\phi,a)\allowbreak\max_{\tilde{a}\in\mathbb{A}_{\tilde{\phi}}} Q(\tilde{\phi},\tilde{a})$, is approximated by our noisy observation $\{\phi,a,\phi'\}$, i.e., $\max_{a'\in\mathbb{A}_{\phi'}}Q_{k}(\phi',a')$. Consequently, the $Q$-value is updated as
\begin{equation}\label{eq:q learning}
\begin{aligned}
Q_{k+1}(\phi,a)=Q_{k}(\phi,a)&+\xi_{k}[r(\phi,a)+\max_{a'\in\mathbb{A}_{\phi'}}Q_{k}(\phi',a')\\
&-Q_{k}(\phi,a)-Q_k(\phi_f,a_f)],
\end{aligned}
\end{equation}
where $(\phi_f,a_f)\in\mathbb{S}\times\mathbb{A}_{\phi_f}$ is an arbitrary but fixed state-action pair, and $\xi_k$ is the step size.}

Under mild assumptions on the step size\footnote{Common choices include $1/k$, ${1}/[k\log (k)]$, and ${\log (k)}/{k}$ for $k\geq 2$.}, the scheme \eqref{eq:q learning} is guaranteed to converge~\cite{abounadi2001learning}. However, the convergence rate may be slow since only one state-action pair $(\phi,a)$ is updated per observation.
To address this limitation, we propose leveraging the structural results established in Lemmas~\ref{lemma:monototic q} and~\ref{lemma:super}. These properties enable simultaneous updates of multiple state-action pairs during each iteration while preserving convergence guarantees, thereby accelerating the process.
\subsection{Structural Update}
The key insight is that Lemmas~\ref{lemma:monototic q} and~\ref{lemma:super} establish relationships between the $Q$-factors. Such ``additional'' information enables the simultaneous update of multiple state-action pairs. The following example illustrates this idea intuitively.
\begin{exmp}{\rm
Suppose a new observation $(\phi_-,a)$ arrives at time $k$ and updating $Q(\phi_-,a)$ via~\eqref{eq:q learning} results in $Q(\phi_-,a)>Q(\phi_+,a)$, where $\phi_-\preceq_{\tau}\phi_+$. Then by {Lemma}~\ref{lemma:monototic q}, this necessitates an increase in $Q (\phi_+,a)$ as well, which effectively allows for two updates from a single observation $(\phi_-,a)$.
}\end{exmp}

In the following, we formally introduce this approach in three steps, as illustrated in Fig.~\ref{fig:relation}. For clarity, we consider a simplified scenario with $N=2$, $l=1$, and $L=1$, and assume there are only two power levels, i.e., $\{0,1\}$. Note that the subsequent analysis can be easily extended to a general case.

For notation simplicity, we denote the $Q$-factor as $Q(^{\tau^{(1)}}_{\tau^{(2)}}~^{a^{(1)}}_{a^{(2)}})$, and stack them as a vector, i.e.,
\begin{figure}[!tbp]
	\centering
	\includegraphics[width=\linewidth]{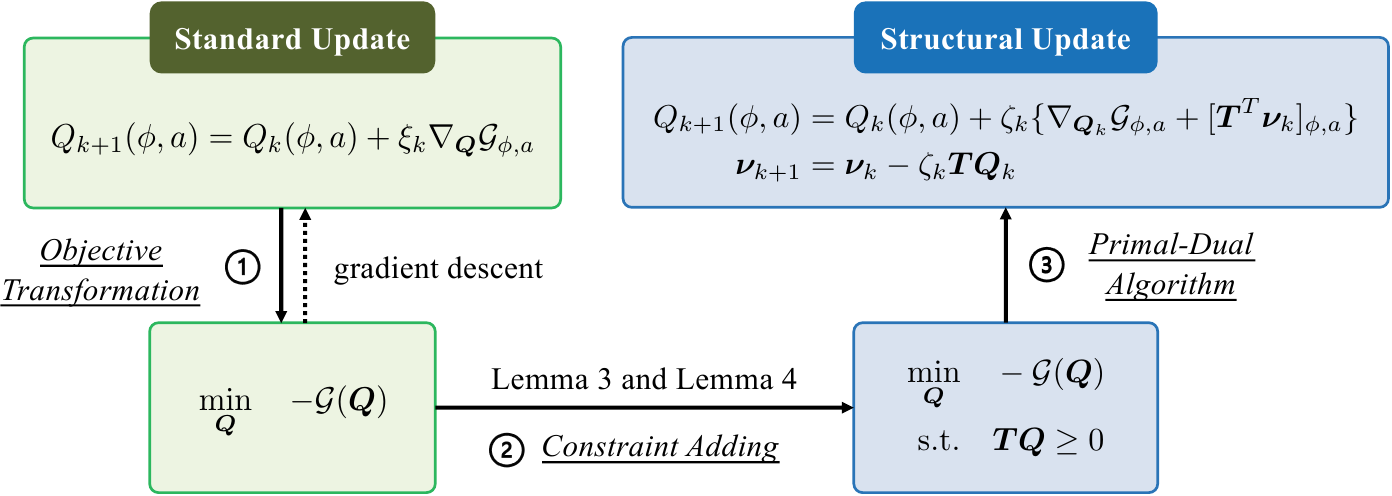}
	\caption{The relationship between the standard update and our proposed structural update.}
	\label{fig:relation}
\end{figure}
\begin{equation*}
\bm{Q}=\left[Q(^0_0~^0_0),Q(^0_0~^0_1),Q(^0_0~^1_0),Q(^0_0~^1_1), Q(^0_1~^0_0),\dots,Q(^1_1~^1_1)\right].
\end{equation*}
\subsubsection{Objective Transformation}
The update scheme~\eqref{eq:q learning} can be viewed as the gradient descent algorithm for an unconstrained optimization problem. To formalize this connection, we consider the following optimization problem:
\begin{equation}\label{eq:gradient}
\begin{aligned}
\min_{\bm Q} 	\quad -\mathcal{G}(\bm{Q}),
\end{aligned}
\end{equation}
where $\mathcal{G}(\bm{Q})$ is a differentiable function of $\bm Q$. The corresponding gradient descent update is $Q_{k+1}(\phi,a)=Q_{k}(\phi,a)+\xi_k\nabla_{\bm{Q}}\mathcal{G}_{\phi,a}$, where $\nabla_{\bm{Q}}\mathcal{G}_{\phi,a}$ is the $(\phi,a)$-th component gradient. Remarkably, if the gradient fulfills
\begin{equation*}
\nabla_{\bm{Q}}\mathcal{G}_{\phi,a}=r(\phi,a)+\max_{a'\in\mathbb{A}_{\phi'}}{Q}(\phi',a')-{Q}(\phi,a)-{Q}(\phi_f,a_f),
\end{equation*}
then the scheme~\eqref{eq:q learning} is indeed the gradient descent method for solving problem~\eqref{eq:gradient}.
\subsubsection{Constraint Adding}
{
Now we aim to transform Lemma~\ref{lemma:monototic q} and~\ref{lemma:super} into constraints on $\bm Q$ that can be incorporated into~\eqref{eq:gradient}.

From Lemma~\ref{lemma:monototic q}, we immediately obtain $Q(^0_0~^0_0) \leq Q (^0_1~^0_0)$, which can be equivalently written in linear form, i.e., $[
	-1	, 0 	, 0	,0, 1,		\bm{0}_{1\times 11}]\bm{Q}\geq 0$.
By enumerating all possible pairs and stacking the resulting inequalities, we obtain the compact matrix form $\bm{T}_m\bm{Q}\geq 0$, where $\bm{T}_m=[\bar{\bm{T}}^T_{m,1},\bar{\bm{T}}^T_{m,2},\bar{\bm{T}}^T_{m,3},\bar{\bm{T}}^T_{m,4}]^T$, $\bar{\bm{T}}_{m,1}=[\bar{\bm T}_m,\bm{0}_{4\times 3}]$, $\bar{\bm{T}}_{m,2}=[\bm{0}_{4\times 1},\bar{\bm T}_m,\bm{0}_{4\times 2}]$, $\bar{\bm{T}}_{m,4}=[\bm{0}_{4\times 3},\bar{\bm T}_m]$ and $\bar{\bm{T}}_{m,3}=[\bm{0}_{4\times 2},\bar{\bm T}_m,\bm{0}_{4\times 1}]$, with $\bar{\bm{T}}_m$ defined in \eqref{eq:tm}.

Similarly, Lemma~\ref{lemma:super} can be represented as $\bm{T}_s\bm{Q}\geq 0$, where $\bm{T}_s$ is defined in \eqref{eq:tm}.
}

\begin{figure*}[!htbp]
{\small
\setcounter{MaxMatrixCols}{20}
\begin{equation}\label{eq:tm}
\bar{\bm{T}}_m=
\begin{bmatrix}
-1	& 0 	& 0	&0& 1		& 0&0		&0&0		&0 & 0 &0& 0 \\
-1	& 0 	& 0	&0& 0		& 0&0		&0&1		&0 & 0 &0& 0 \\
0	& 0 	& 0	&0& -1		& 0&0		&0&0		&0 & 0 &0& 1 \\
0	& 0 	& 0	&0& 0		& 0&0		&0&-1		&0 & 0 &0& 1 
\end{bmatrix},\quad
\bm{T}_s=
\begin{bmatrix}
1 & -1 & 0  & 0  & -1 & 1 & 0 & 0  & 0  & 0 & 0  & 0 & 0  & 0  & 0  & 0 \\
0 & 0  & 1  & -1 & 0 & 0 & -1 & 1  & 0  & 0 & 0  & 0 & 0  & 0  & 0  & 0 \\
0 & 0  & 0  & 0  & 0 & 0 & 0 & 0  & 1  & -1 & 0  & 0 & -1 & 1  & 0  & 0 \\
0 & 0  & 0  & 0  & 0 & 0  & 0 & 0 & 0  & 0 & 1  & -1 & 0  & 0 & -1  & 1 \\
1 & 0  & -1 & 0  & 0 & 0  & 0 & 0  & -1 & 0 & 1  & 0 & 0  & 0  & 0  & 0 \\
0 & 1  & 0  & -1 & 0 & 0  & 0 & 0  & 0  & -1 & 0 & 1 & 0  & 0  & 0  & 0 \\
0 & 0  & 0  & 0  & 1 & 0  & -1 & 0  & 0  & 0 & 0 & 0 & 0 & -1 & 0  & 1 \\
0 & 0  & 0  & 0  & 0 & 0  & 0 & 0  & 0  & 0 & -1 & 1 & 0  & 0  & -1 & 1
\end{bmatrix}.
\end{equation}}
\end{figure*}

This formulation allows us to interpret the learning scheme as constrained optimization:
\begin{equation}\label{eq:gradient new}
\min_{\bm Q} 	\quad -\mathcal{G}(\bm{Q}),\qquad\qquad{\rm s.t.}		\quad \bm{T}\bm{Q}\geq0,
\end{equation}
where $\bm{T}=[\bm{T}_m;\bm{T}_s]$. 
\subsubsection{Primal-Dual Algorithm}
Now our task becomes solving the problem~\eqref{eq:gradient}. Since $\nabla^2(-\mathcal{G}(\bm Q))\succeq0$ and the constraint is linear, \eqref{eq:gradient} is convex. Consequently, the following primal-dual algorithm can be applied to solve the problem:
\begin{align}
{Q}_{k+1}(\phi,a)	&={Q}_k(\phi,a)+\xi_k\{\nabla_{\bm {Q}_k}\mathcal{G}_{\phi,a}+[\bm{T}^T\bm{\nu}_k]_{\phi,a}\},\label{eq:primal}\\
\bm{\nu}_{k+1}		&=\max\{\bm{\nu}_k-\xi_k\bm{T}\bm{Q}_k,0\}.\label{eq:dual}
\end{align}
%
It is guaranteed to converge to the solution of the Bellman equation, as formalized in the following theorem.
\begin{thm}\label{thm:converge}{\rm
The structural update rules~\eqref{eq:primal}-\eqref{eq:dual} converge to a solution to~\eqref{eq:q update} almost surely.
}\end{thm}
\begin{pf}
See Appendix~\ref{apx:converge}.
\end{pf}
{
\floatname{algorithm}{Algorithm}
\begin{algorithm}[!htbp]
\caption{Structural Stochastic Approximation for Unknown Channel Model}
\label{alg:structural_q}

\begin{algorithmic}[1]
\State Initialize $\bm Q_0$ and $\bm\nu_0$ arbitrarily and set $k\gets 0$
\Repeat
    \State Observe current state $\phi_k$
    \State Select action $a_k$ using the $\epsilon$-greedy rule~\eqref{eq:e_greedy}
    \State Execute $a_k$, and observe $r(\phi_k,a_k)$ and $\phi_{k+1}$
    \For{each state-action pair $(\phi,a)$}
        \State Compute $\bm{Q}_{k+1}$ from $\bm{Q}_{k}$ and $\bm\nu_{k}$ via~\eqref{eq:primal}
    \EndFor
    \State Compute the dual variable $\bm \nu_{k+1}$ via~\eqref{eq:dual}
    \State $k\gets k+1$
\Until{convergence}
\State \Return $\theta(\phi)\gets \arg\max_{a\in\mathbb{A}_\phi}Q_k(\phi,a)$
\end{algorithmic}
\end{algorithm}}
{A summary of the procedure for computing $\bm Q_k$ and thus, the optimal policy is provided in Algorithm~\ref{alg:structural_q}.\footnote{The per-iteration complexity for the standard update is on the order of $\mathcal{O}(|\mathbb{A}|)$. Compared with the standard update, Algorithm~\ref{alg:structural_q} involves additional matrix multiplications $\bm{T}\bm{\nu}_k$ and $\bm{T}\bm{Q}_k$. Nevertheless, the sparsity pattern of $\bm{T}$ shown in~\eqref{eq:tm} ensures that the per-iteration overhead remains minimal. The primary benefit of Algorithm 2 is increased sample efficiency, which accelerates convergence by reducing the total number of required iterations.}}
\begin{rem}{\rm
Compared to the unconstrained optimization problem~\eqref{eq:gradient}, the problem \eqref{eq:gradient new} is now constrained. Intuitively, this means that the feasible domain is reduced at each time, making the corresponding learning scheme~\eqref{eq:primal}-\eqref{eq:dual} more efficient than~\eqref{eq:q learning}. A visual illustration of the relationship between the standard update~\eqref{eq:q learning} and our proposed structural update~\eqref{eq:primal}-\eqref{eq:dual} is provided in Fig.~\ref{fig:relation}.
}\end{rem}
\section{Simulations}\label{chap:simulation}
This section presents simulations to validate the theoretical results and highlight the efficacy of the proposed algorithms. We illustrate the threshold-based behavior of the optimal policy and demonstrate the superior performance of our proposed algorithm compared to heuristic algorithms under perfect channel information. Next, we extend the analysis to unknown channel models, where structural updates achieve faster convergence and higher rewards during training compared to standard methods. The results collectively confirm the monotonicity and superadditivity properties derived earlier while highlighting the efficacy of the proposed algorithms.
{
\subsection{Simulation Setup: Power Systems}\label{chap:sim setup}
In this section, we consider two practical power-network examples: a 4-bus system~\cite{grainger2018power} and a modified PJM 5-bus system~\cite{li2010small}. The linearized DC power flow measurement models are generated using MATPOWER~\cite{zimmerman2010matpower}, where $C_1$ and $C_2$ correspond to the linearized DC power flow Jacobians relating the bus phase-angle states to the power-flow measurements. Moreover, $Q_1$, $R_1$, $Q_2$ and $R_2$ are chosen as identity matrices.

We consider two cases for $A$ matrices. For Case 1, we adopt the commonly used random-walk model, where, $A_1 = I$ and $A_2=I$. For Case 2, leakage and damping effects are incorporated into the system dynamics. 

\begin{rem}{\rm
		It is worth noting that, in Case 2, $\|A_1\|= 0.9936$ and $\|A_2\|=0.9775$. By considering both cases, we cover both scenarios where $\|A_i\|\geq 1$ and $\|A_i\|< 1$.
	}\end{rem}

Unless otherwise specified, the variable $\tau^{(i)}$ is restricted to a finite set $\{0,\dots,20\}$, and the transmission power is normalized with $\mathcal{P}=\{0,1,2\}$. The channels gains $H_k^{(i)}$, and $G_k^{(i)}$ have two possible states: $s_1 = 0.04$ and $s_2 = 0.09$, with transition probabilities $t^s_{s_1,s_1}=t^s_{s_2,s_2}=0.8$ and $t^s_{s_1,s_2}=t^s_{s_2,s_1}=0.2$. For the modulation scheme, we employ QAM (Example~\ref{exam:qam}) with parameters $\alpha=0.75$, $b=0.8$ and $B=4$. The noise level is set to $\sigma_i=0.1$. For the energy harvester, we consider $b_{\max} = 3$ and $\mathcal{E}=\{0,1,2\}$, where the transition probability is provided below~\cite{leong2018transmission}: $t^e_{0,0}=t^e_{2,1}=0.2$, $t^e_{0,1}=t^e_{1,0}=t^e_{1,2}=0.3$, $t^e_{0,2}=0.5$, $t^e_{1,1}=0.4$, $t^e_{2,0}=0.1$, and $t^e_{2,2}=0.7$.

\subsection{Results for Perfect Channel Information}
1) {\it\bf Structural Results:} Fig.~\ref{fig:monotonic} presents the optimal policy at specific values of $(\tau^{(1)},\tau^{(2)})$ under conditions $b=2$ and $E=1$, for different channel realizations. We observe clear action transitions between a certain threshold and the action $a^{(i)}$ is non-decreasing in $\tau^{(i)}$, which confirms the structural results established in Theorem~\ref{thm:structure}.

\begin{figure}[!htbp]
	\centering
	\includegraphics[width=0.9\linewidth]{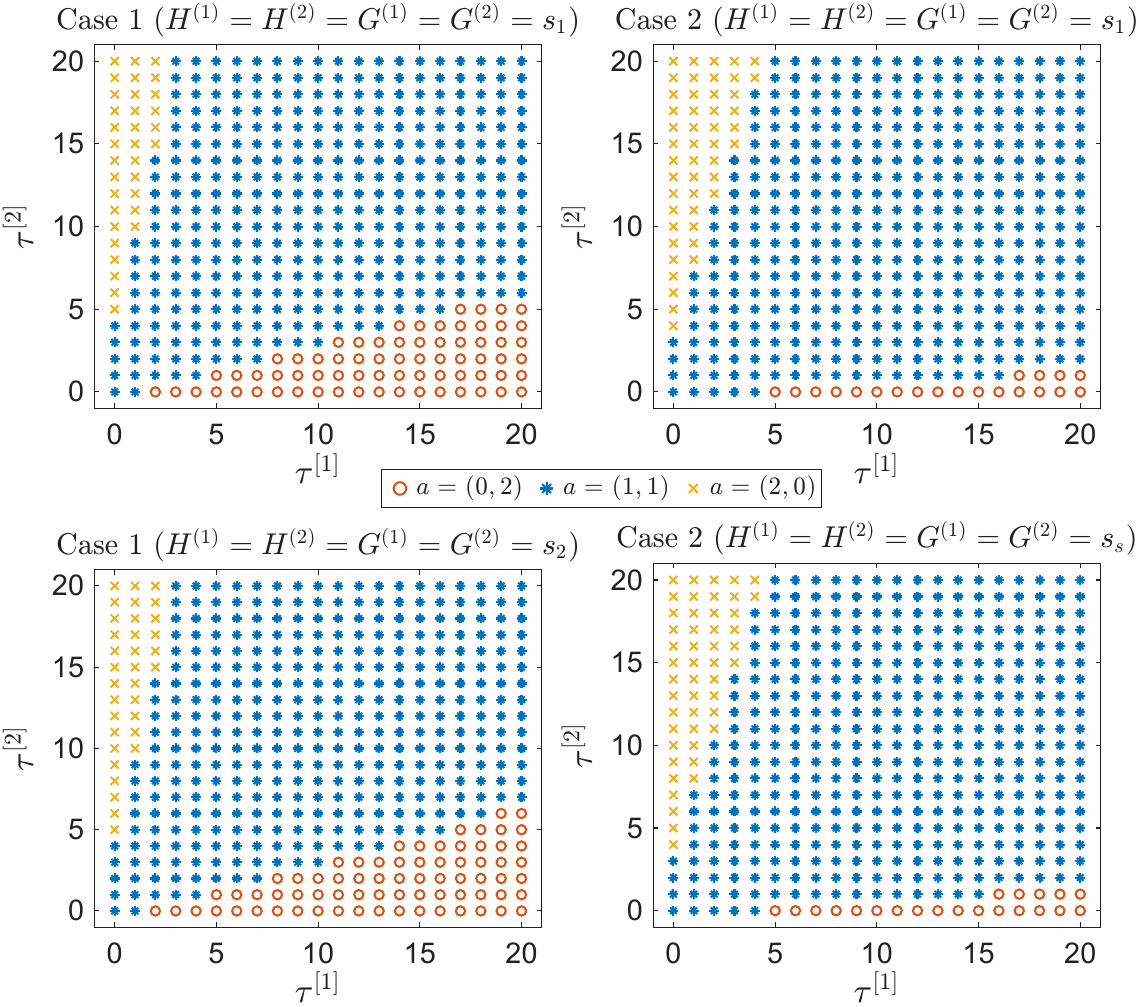}
	\caption{The optimal policy at $(\tau^{(1)},\tau^{(2)})$.}
	\label{fig:monotonic}
\end{figure}

\begin{figure*}[!t]
	\centering
	\includegraphics[width=0.8\linewidth]{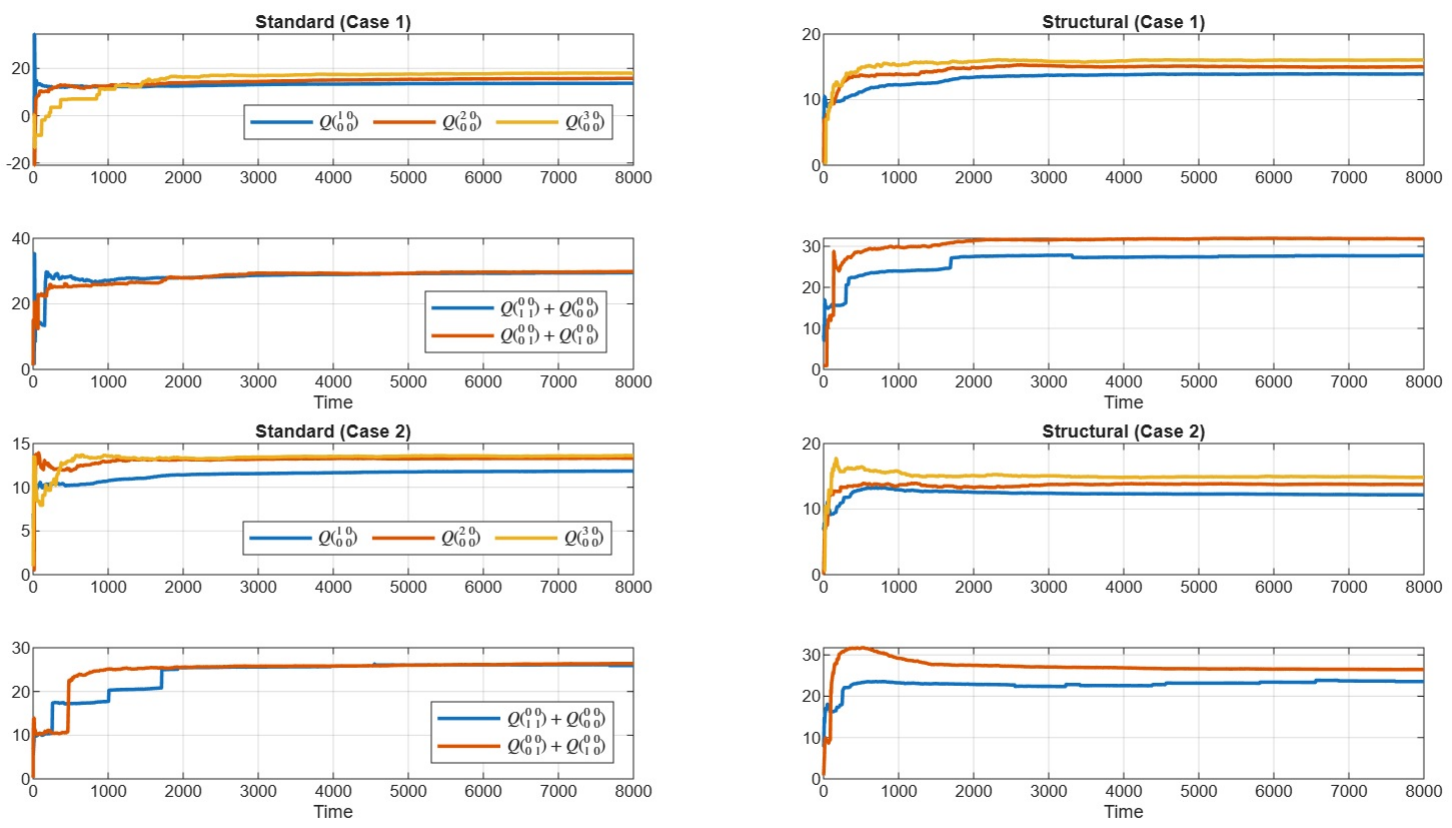}
	\caption{$Q$-factor in the learning process with standard and structural updates. Monotonicity and superadditivity are satisfied when the blue line is the lowest, the yellow line is the largest and the orange line lies in between.}
	\label{fig:Q}
\end{figure*}

\begin{figure}[!tbp]
	\centering
	\includegraphics[width=0.9\linewidth]{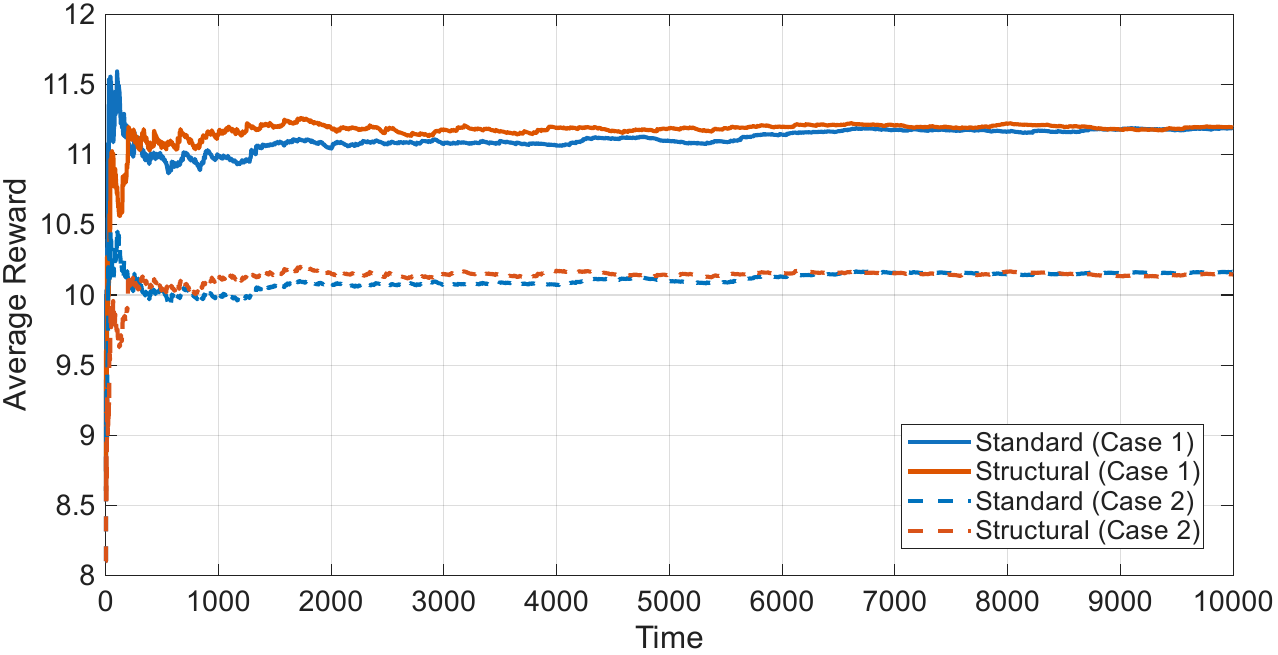}
	\caption{Average reward in the learning process with standard and structural updates.}
	\label{fig:reward}
\end{figure}

2) {\it\bf Policy Comparison}: To demonstrate the optimality of the relative value iteration algorithm, we compare its performance against two alternative strategies: a greedy algorithm and a random policy.  In the greedy algorithm, the attacker allocates the maximum available power to the channel with the largest $\tau^{(i)}$, then proceeds to the the-largest $\tau^{(i)}$, and so forth, until the battery is exhausted. The average rewards achieved by these strategies are presented in Table~\ref{tb:compare}. It is evident that relative value iteration achieves the highest reward. Moreover, Theorem~\ref{thm:structure} suggests that allocating more power to channels with larger $\tau^{(i)}$ is beneficial. Consistent with this insight, the greedy algorithm exhibits a substantial performance improvement compared with the random baseline.
\begin{table}[!htbp]
\centering
\caption{Average reward obtained by different policies.}
\label{tb:compare}

\begin{tabular}{cccc}
\hline
~& Value Iteration & Greedy & Random \\ \hline
Case 1&      11.38  &  10.93     & 9.60     \\
Case 2&      10.31  &  9.61     & 9.14     \\ \hline
\end{tabular}
\end{table}

\subsection{Results for Unknown Channel Model}
1) {\it\bf Structural Update}: To better visualize the learning process of the $Q$-factor, we restrict our focus to the case where the communication channel is fixed and $\tau^{(i)}\in\{0,\dots,7\}$. Fig.~\ref{fig:Q} illustrates the evolution of the $Q$-factor under standard~\eqref{eq:q learning} and structural~\eqref{eq:primal}-\eqref{eq:dual} updates. While both approaches converge, the structural version demonstrates faster convergence. Furthermore, Fig.~\ref{fig:Q} reveals that the standard updates do not consistently satisfy the monotonicity (Lemma~\ref{lemma:monototic q}) and superadditivity (Lemma~\ref{lemma:super}) properties. In contrast, the structural updates largely preserves these properties, which accounts for its improved convergence behavior.

2) {\it\bf Average Rewards}: Fig.~\ref{fig:reward} shows the average reward for the setup described in Section~\ref{chap:sim setup}. While the final average rewards of both algorithms converge to the same value, the structural updates yield higher rewards during the learning process compared with the standard one.}

\section{Conclusion}\label{chap:conclusion}
In this paper, we analyzed jamming attacks against remote state estimation. We aimed to develop algorithms to determine optimal jamming policies in the context of the ``harvest and jam'' scenario. Two cases were studied based on the attacker's knowledge of the system: (i) perfect channel knowledge and (ii) unknown channel model. We demonstrated that the proposed algorithms for both cases converge to the optimal policy asymptotically. Furthermore, the optimal policy was shown to exhibit certain structural properties that can be leveraged to expedite both algorithms. Simulation results validated the optimality of the obtained policies. Additionally, the accelerated version of the algorithms exhibited faster convergence and yielded higher average rewards during the learning process compared to the standard algorithm. Future work includes {considering energy harvesting from legitimate communications between sensors and the estimator} or exploring a game-theoretic framework where sensors are aware of the energy-harvesting attacker.
\bibliographystyle{plain}        
\bibliography{ref}           
\appendix
\section{Proof of Theorem~\ref{thm:stationary}}\label{apx:stationary}
Considering a policy $a_k=\theta(\phi_k)$, we can express its discounted reward as follows:
\begin{equation}\label{eq:discounted cost}
V_{\delta}(\theta,\phi)\triangleq\sum_{k=0}^{\infty}\delta^k\mathbb{E}[r(\phi_k,\theta(\phi_k))|\phi_0=\phi],
\end{equation}
where $\delta\in(0,1)$ is the discount factor. Denote $V_{\delta}(\phi)\triangleq\sup_{\theta\in{\Theta}}V_{\delta}(\theta,\phi)$.

According to to~\cite[Theorem 4.1, Remark 3.2, Remark 4.1(b)]{guo2006average}, only the following items need to be verified:
\begin{enumerate}[(I)]
\item There exist positive constants $\zeta$ and $\mu<1$ and a measurable function $\Psi(\cdot)\geq1$ on $\mathbb{S}$ such that for all $(\phi,a)\in\mathbb{S}\times\mathbb{A}$, we have $\sum_{\phi'\in\mathbb{S}}\Psi(\phi')\mathbb{P}(\phi'|\phi,a)\leq\mu \Psi(\phi)+\zeta$.
\item For all $(\phi,a)\in\mathbb{S}\times\mathbb{A}$, there exists a constant $M\geq0$ such that $|r(\phi,a)|\leq M\Psi(\phi)$.
\item The action space $\mathbb{A}$ is finite.
\item The state space $\mathbb{S}$ is countably infinite.
\item There exists a state $\phi_f\in\mathbb{S}$ and two functions $l(\cdot)$ and $u(\cdot)$ on $\mathbb{B}_{\Psi}(\mathbb{S})$ such that $-l(\phi)\leq V_{\delta}(\phi)-V_{\delta}(\phi_f)\leq u(\phi)$, for all $\phi\in\mathbb{S}$ and $\delta\in(0,1)$, where $\mathbb{B}_{\Psi}(\cdot)$ is the Banach space given $\Psi(\cdot)$.
\end{enumerate}
We choose 
\begin{equation}\label{eq:measure}
\Psi(\phi)		=\sum_{i=1}^N\Psi(\tau^{(i)}), \quad \Psi(\tau^{(i)})	=\|A_i\|^{2\tau^{(i)}}.
\end{equation}
Since $\|A_i\|\geq1$, $\Psi(\phi)\geq1$ for any $\phi\in\mathbb{S}$. Then when the state $\phi$ transitions to $\phi'$, the function $\Psi(\phi'^{(i)})$ is updated according to the following rule: $\Psi(\tau'^{(i)})=1$ if $\tau'^{(i)}=0$ and $\Psi(\tau'^{(i)})=\|A_i\|^2\Psi(\tau^{(i)})$ if $\tau'^{(i)}=\tau^{(i)}+1$.
\begin{figure}[!htbp]
	\centering
	\includegraphics[width=\linewidth]{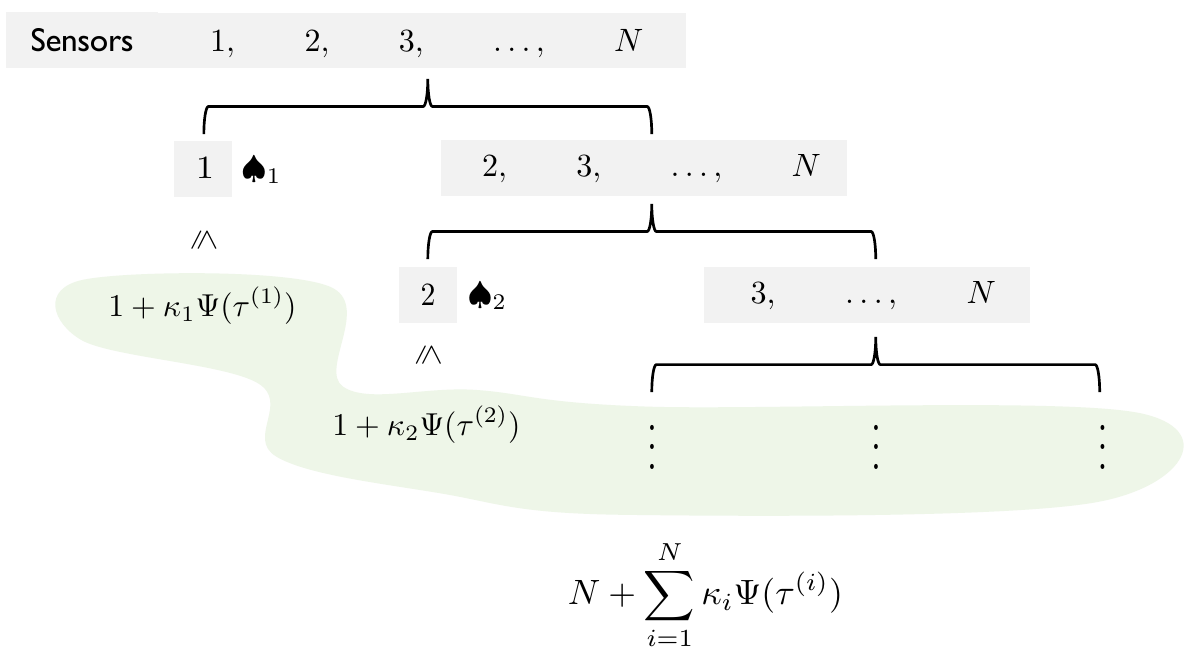}
	\caption{Illustration of the derivation process of~\eqref{eq:derivation}.}
	\label{fig:proof}
\end{figure}

For ease of exposition, we focus on the case with $N=2$ sensors without loss of generality. We calculate the value of $\sum_{\phi'\in\mathbb{S}}\Psi(\phi')\mathbb{P}(\phi'|\phi,a)$ by fixing $\Lambda_{2}$ while varying $\Lambda_{1}$. Through manipulations and relaxations (Fig.~\ref{fig:derivation} and Fig.~\ref{fig:proof}), the calculation is reduced to the case where only the $2$-nd sensor is present. Repeating this process for each sensor in the general case of $N$ sensors, we obtain
\begin{equation}\label{eq:derivation}
\sum_{\phi'\in\mathbb{S}}\Psi(\phi')\mathbb{P}(\phi'|\phi,a)\leq+\max_i{\kappa_i}\Psi(\phi).
\end{equation}
Let $\mu=\max_i{\kappa_i}<1$ and $\zeta=N$, (I) is proved.
\begin{figure*}[!t]
	\begin{tcolorbox}[enhanced jigsaw,
		colback=gray!5,
		colframe=gray,
		width=\linewidth,%
		arc=2mm, auto outer arc,
		boxrule=1.5pt,
		drop shadow={gray}
		]
		We begin with the following expression:
		\begin{equation*}
			\begin{aligned}
				&\sum_{\phi'\in\mathbb{S}}\Psi(\phi')\mathbb{P}(\phi'|\phi,a)\\
				&=\sum_{\Lambda_{2}'}\underbrace{\left[\sum_{\Lambda_1'}t^s_{F^{(1)},F'^{(1)}}t^s_{G^{(1)},G'^{(1)}}(\lambda^{(1)})^{\gamma'^{(1)}}(1-\lambda^{(1)})^{1-\gamma'^{(1)}}\Psi(\phi')\right]}_{\spadesuit_1} t^s_{F^{(2)},F'^{(2)}}t^s_{G^{(2)},G'^{(2)}}(\lambda^{(2)})^{\gamma'^{(2)}}(1-\lambda^{(2)})^{1-\gamma'^{(2)}},
			\end{aligned}
		\end{equation*}
		where
		\begin{equation*}
				\spadesuit_1=
				\sum_{\Xi'_1}t_{F^{(1)},F'^{(1)}}t_{G^{(1)},G'^{(1)}}\left[\lambda^{(1)}\left(1+\Psi(\tau'^{(2)})\right)+(1-\lambda^{(1)})\left(\|A_1\|^2\Psi(\tau^{(1)})+\Psi(\tau'^{(2)})\right)\right].
		\end{equation*}
		Using the property $\lambda^{(i)}\leq 1$ and Assumption~\ref{assumption:stable}, we can bound ${\spadesuit_1}$ as follows:
		\begin{equation*}
			\begin{aligned}
				\spadesuit_1&\leq1+\sum_{\Xi_1}t_{F^{(1)},F'^{(1)}}t_{G^{(1)},G'^{(1)}}(1-\lambda^{(1)})\|A_1\|^2\Psi(\tau^{(1)})+\Psi(\tau'^{(2)})\leq1+\kappa_1\Psi(\tau^{(1)})+\Psi(\tau'^{(2)}),
			\end{aligned}
		\end{equation*}
		
		where $\kappa_1$ is defined in Assumption~\ref{assumption:stable}.
		
		Therefore, substituting ${\spadesuit_1}$ back into the original equation for $\sum_{\phi'\in\mathbb{S}}\Psi(\phi')\mathbb{P}(\phi'|\phi,a)$, we have
		\begin{equation*}
			\begin{aligned}
				\sum_{\phi'\in\mathbb{S}}\Psi(\phi')\mathbb{P}(\phi'|\phi,a)&\leq1+\kappa_1\Psi(\tau^{(1)})+\underbrace{\sum_{\Lambda_2'}t_{F^{(2)},F'^{(2)}}t_{G^{(2)},G'^{(2)}}(\lambda^{(2)})^{\gamma'^{(2)}}(1-\lambda^{(2)})^{1-\gamma'^{(2)}}\Psi(\tau'^{(2)})}_{\spadesuit_2}\\
				&\leq2+\kappa_1\Psi(\tau^{(1)})+\kappa_2\Psi(\tau^{(2)}).
			\end{aligned}	
		\end{equation*}
	\end{tcolorbox}
	\caption{Detailed derivation process of~\eqref{eq:derivation} with 2 sensors.}
	\label{fig:derivation}
\end{figure*}

Let $\varphi$ be a constant such that $\bar{P}_i\preceq\varphi I$ and $W_i\preceq\varphi I$. Define a function $\bar{h}_i(X)\triangleq A_iXA_i^T+\varphi I$. Then $h^{\tau^{(i)}}_i(\bar{P}_i)\preceq \bar{h}^{\tau^{(i)}}_i(\varphi I)\preceq\varphi\sum_{j=0}^{\tau^{(i)}}A_i^j(A_i^T)^j$.
Therefore,
\begin{equation*}
\begin{aligned}
&\frac{\tr[h^{\tau^{(i)}+1}_i(\bar{P}_i)]}{\varphi}	\leq n_i+n_i\|A_i\|^2+\cdots+n_i\|A_i\|^{2\tau^{(i)}+2}\\
						&=n_i\frac{\|A_i\|^{2\tau^{(i)}+4}-1}{\|A_i\|^2-1}\leq\frac{n_i\|A_i\|^4}{\|A_i\|^2-1}\Psi(\tau^{(i)}).
\end{aligned}
\end{equation*}
Therefore, we obtain ${\tr[h^{\tau^{(i)}+1}_i(\bar{P}_i)]}\leq M\Psi(\tau^{(i)})$, where $M=\varphi\max_i\frac{ n_i\|A_i\|^4}{\|A_i\|^2-1}$.
Then, (II) is proved since:
\begin{equation*}
\begin{aligned}
0&\leq r(\phi,a)	=\sum_{i=1}^N\left\{\lambda^{(i)}\tr[\bar{P}_{i}]+(1-\lambda^{(i)})\tr[h_i^{\tau^{(i)}+1}(\bar{P}_{i})]\right\}\\
			&\leq\sum_{i=1}^N\tr[h_i^{\tau^{(i)}+1}(\bar{P}_{i})]\leq\sum_{i=1}^N M\Psi(\tau^{(i)})=M\Psi(\phi).
\end{aligned}
\end{equation*}
According to~\cite[Lemma 3.2]{guo2006average}, for any $\delta\in(0,1)$ and state $\phi\in\mathbb{S}$, the following holds:
\begin{equation*}
V_{\delta}(\phi)\leq \frac{M\Psi(\phi)}{1-\delta}+\frac{M\zeta}{(1-\mu)(1-\delta)}.
\end{equation*}
Let $\phi_f=(b, E, \Lambda_f^{(1)},\dots,\Lambda_f^{(N)})$, where $\Lambda_f^{(i)}=(H^{(i)},G^{(i)},\allowbreak 0)$. Since $V_{\delta}(\phi)\geq0$, we obtain
\begin{equation*}
\begin{aligned}
&V_{\delta}(\phi)-V_{\delta}(\phi_f)	\geq0-V_{\delta}(\phi_f)\\&	\geq  -\frac{M\Psi(\phi_f)}{1-\delta}-\frac{M\zeta}{(1-\mu)(1-\delta)}=-\frac{(N-\mu N+\zeta)M}{(1-\mu)(1-\zeta)}.
\end{aligned}
\end{equation*}
Let $\underline{k}=\inf_{k\geq0}\{k|\phi_k=\phi_f,\phi_0=\phi\}$. Then
\begin{align*}
&V_{\delta}(\phi)	=\mathbb{E}\left[\sum_{k=0}^{\underline{k}}\delta^kr(\phi_k,\theta^*(\phi_k))|\phi_0=\phi\right]\\
				&\quad+\mathbb{E}\left[\sum_{k=\underline{k}}^{\infty}\delta^{k}r(\phi_k,\theta^*(\phi_k))|\phi_0=\phi,\phi_{\underline{k}}=\phi_f\right]\\
				&=\mathbb{E}\left[\sum_{k=0}^{\underline{k}}\delta^kr(\phi_k,\theta^*(\phi_k))|\phi_0=\phi\right]+\mathbb{E}[\delta^{\underline{k}}|\phi_0=\phi]V_{\delta}(\phi_f)\\
				&\leq\mathbb{E}\left[\sum_{k=0}^{\underline{k}}\delta^kr(\phi_k,\theta^*(\phi_k))|\phi_0=\phi\right]+V_{\delta}(\phi_f).
\end{align*}
By Lemma~\ref{lemma:bound}, $\mathbb{E}[\sum_{k=0}^{\underline{k}}\delta^kr(\phi_k,\theta^*(\phi_k))|\phi_0=\phi]$ is bounded. Therefore, by letting $l(\phi)=(N-\mu N+\zeta)M/[(1-\mu)(1-\zeta)]$ and $u(\phi)	=\mathbb{E}[\sum_{k=0}^{\underline{k}}\delta^kr(\phi_k,\theta^*(\phi_k))|\allowbreak\phi_0=\phi]$, we complete the proof of (V), and thus the proof of Theorem~\ref{thm:stationary}.
\section{Proof of Theorem~\ref{thm:structure}}\label{apx:structure}
Using Lemma~\ref{lemma:super}, we conclude that if $ Q (\tau^{(i)}_-,a_{+,i})\geq Q (\tau^{(i)}_-,a_{-,i})$, then $ Q (\tau^{(i)}_+,a_{+,i})\geq Q (\tau^{(i)}_+,a_{-,i})$. This result implies that if, under $\tau^{(i)}_-$, action $a_{+,i}$ yields a higher reward than action $a_{-,i}$, then the same holds true for $\tau^{(i)}_+$. This is exactly what we aim to prove.

{
\section{Proof of Theorem~\ref{thm:converge}}\label{apx:converge}
The proof of Theorem~\ref{thm:converge} follows a similar trajectory to the convergence proof of relative value iteration (RVI) $Q$-learning presented by Abounadi~\cite{abounadi2001learning}. The key difference is that we must additionally establish the asymptotic stability of the structural update rules~\eqref{eq:primal}-\eqref{eq:dual}.

Define $\hat{\mathcal{H}}_{\phi,a}(\bm Q)\triangleq r(\phi,a)-Q(\phi,a)-Q(\phi_f,a_f)+\max_{a'\in\mathbb{A}_{\phi'}}Q(\phi',a')$, and $\mathcal{H}_{\phi,a}(\bm Q)\triangleq r(\phi,a)-Q(\phi,a)-Q(\phi_f,a_f)+\sum_{\tilde{\phi}\in\mathbb{S}}\mathbb{P}(\tilde{\phi}|\phi,a)\max_{\tilde{a}\in\mathbb{A}_{\tilde{\phi}}}Q(\tilde{\phi},\tilde{a})$.

As discussed in Section~\ref{sec:std_update}, each update replaces the expectation in $\mathcal{H}_{\phi,a}$ with our noisy observation in $\hat{\mathcal{H}}_{\phi,a}$. Then,~\eqref{eq:q learning} be written as the stochastic approximation
\begin{equation}\label{eq:stochastic approx}
\begin{aligned}
Q_{k+1}(\phi,a)=Q_{k}(\phi,a)&+\xi_{k}[\mathcal{H}_{\phi,a}(\bm Q_k)+M_{k+1}],
\end{aligned}
\end{equation}
where $\mathcal{H}_{\phi,a}(\bm Q_k)$ corresponds to the ideal update with known transition probabilities, and $M_{k+1}$ is the resulting noise term by our sample update.

Note  that~\eqref{eq:stochastic approx} corresponds to a limiting ODE $\dot{\bm Q}=\mathcal{H}(\bm Q)$. By~\cite[Theorem 3.4]{abounadi2001learning}, the RVI Q-learning updates converge almost surely to a solution to~\eqref{eq:q update} since:
\begin{enumerate}[(I)]
\item The noise sequence $\{M_k\}$ is a square-integrable martingale difference sequence
\item The limiting ODE $\dot{\bm Q}=\mathcal{H}(\bm Q)$ has a unique asymptotically stable equilibrium.
\end{enumerate}

For the structural updates, (I) holds directly since the observation process remains unchanged. It therefore remains to verify the stability of the associated ODE.

In particular, for the update rules~\eqref{eq:primal}-\eqref{eq:dual}, the corresponding ODE is given by
\begin{equation}\label{eq:ode}
\dot{\bm Q}	=\mathcal{\bm H}(\bm Q)+\bm{T}^T\bm{\nu},\qquad \dot{\bm{\nu}}	=\max\{\bm\nu-\bm{TQ},0\}-\bm\nu.
\end{equation}
Let $\hat{\bm Q}$ denote the equilibrium of $\dot{\bm Q}=\mathcal{H}(\bm Q)$, so that $\mathcal{H}(\hat{\bm Q})=0$ and $\bm{T} \hat{\bm Q}\geq 0$. To establish the convergence of~\eqref{eq:ode}, we construct the Lyapunov function $\mathcal{V}(\bm{Q},\bm{\nu})\triangleq\frac{1}{2}(\bm{Q}-\hat{\bm{Q}})^T(\bm{Q}-\hat{\bm{Q}})+\frac{1}{2}\bm{\nu}^T\bm{\nu}$. Using $\mathcal{H}(\hat{\bm Q})=0$, $\bm{T} \hat{\bm Q}\geq 0$ and $\bm\nu\geq0$, we obtain
\begin{equation}\label{eq:v dot}
\begin{aligned}
\dot{\mathcal{V}}	&=(\bm{Q}-\hat{\bm{Q}})^T(\mathcal{\bm H}(\bm Q)+\bm{T}^T\bm{\nu})-\bm{\nu}^T\bm{TQ}\\
				&=(\bm{Q}-\hat{\bm{Q}})^T\mathcal{\bm H}(\bm Q)-\bm \nu^T \bm T \hat{\bm Q}\\
				&\leq(\bm{Q}-\hat{\bm{Q}})^T[\mathcal{\bm H}(\bm Q)-\mathcal{\bm H}(\hat{\bm Q})].
\end{aligned}
\end{equation}
Moreover, it is easy to verify that the Bellman operator $\mathcal{H}(\bm Q)+\bm Q$ is non-expansive in the $\ell_2$ norm\footnote{More precisely, the non-expansivity of the Bellman operator is established with respect to a weighted $\ell_2$ norm, i.e., $\|\cdot\|_{\bm\Lambda}$ where $\bm\Lambda$ corresponds to the stationary distribution. In this case, the rigorous Lyapunov function takes the form $\mathcal{V}(\bm{Q},\bm{\nu})=\frac{1}{2}(\bm{Q}-\hat{\bm{Q}})^T\bm \Lambda(\bm{Q}-\hat{\bm{Q}})+\frac{1}{2}\bm{\nu}^T\bm{\nu}$. For notational brevity, we proceed using the standard $\ell_2$ norm, as the convergence arguments remain qualitatively identical.}, i.e., $(\bm{Q}-\hat{\bm{Q}})^T[\mathcal{H}(\bm Q)+\bm Q-\mathcal{H}(\hat{\bm Q})-\hat{\bm Q}]\leq \|\bm{Q}-\hat{\bm{Q}}\|_2^2$.

Combining this with the previous inequality~\eqref{eq:v dot} yields $\dot{\mathcal V} \leq 0$, with equality if and only if $(\bm Q, \bm \nu)=(\hat{\bm Q},\bm 0)$. Therefore, from LaSalle’s invariance principle, \eqref{eq:ode} is globally asymptotically stable, and thus~\eqref{eq:primal}-\eqref{eq:dual} converge to the solution to~\eqref{eq:q update} almost surely. 
}
\section{Proof of Lemma~\ref{lemma:monototic}}\label{apx:monototic}
Recall the discounted reward $V_{\delta}(\phi)$ defined in~\eqref{eq:discounted cost}. Similar to the average-cost problem, the optimal policy for the discounted cost problem satisfies the Bellman equation:
\begin{equation}\label{eq:bellman discount}
V_{\delta}(\phi)=\max_{a\in\mathbb{A}_\phi}\{r(\phi,a)+\delta\sum_{\phi'\in\mathbb{S}}\mathbb{P}(\phi'|\phi,a)V_{\delta}(\phi')\}.
\end{equation}
Define the mapping $\mathcal{T}_{\delta}(\cdot)$ as the Bellman operator on $\hat{V}_{\delta}(\phi)$, $\phi\in\mathbb{S}$ as
\begin{equation}\label{eq:map}
\mathcal{T}_{\delta}\hat{V}_{\delta}(\phi)=\max_{a\in\mathbb{A}_\phi}\{r(\phi,a)+\delta\sum_{\phi'\in\mathbb{S}}\mathbb{P}(\phi'|\phi,a)\hat{V}_{\delta}(\phi')\}.
\end{equation}
In what follows, we first demonstrate that $\mathcal{T}_{\delta}(\cdot)$ is a contraction mapping. We then apply the induction-based method to prove that $V_{\delta}(\phi)$ is monotonic, and subsequently show that $V(\phi)$ is also monotonic.
\begin{defn}[Weighted supremum norm]{\rm
Let $\Psi(\cdot)$ be an arbitrary positive real-valued function on $\mathbb{S}$. The weighted supremum norm $\|\cdot\|_\Psi$ for real-valued functions $V(\cdot)$ on $\mathbb{S}$, is defined as $\|V\|_\Psi=\sup_{\phi\in\mathbb{S}}|V(\phi)|/\Psi(\phi)$.
}\end{defn}

\begin{defn}[Contraction mapping]{\rm
The mapping $\mathcal{T}_{\delta}(\cdot)$ is a contraction mapping on $\mathbb{B}_\Psi(\mathbb{S})$ with {respect} to $\Psi(\cdot)$, if for all $\hat{V}_{\delta}, \hat{V}_{\delta}'\in\mathbb{B}_\Psi(\mathbb{S})$, the inequality $\|\mathcal{T}_{\delta}\hat{V}_{\delta}-\mathcal{T}_{\delta}\hat{V}_{\delta}'\|_\Psi\leq\vartheta \|\hat{V}_{\delta}-\hat{V}_{\delta}'\|_\Psi$ holds, where $0\leq\vartheta<1$.
}\end{defn}
\begin{lem}{\rm
If Assumption~\ref{assumption:stable} holds, then the Bellman operator $\mathcal{T}_{\delta}(\cdot)$ in~\eqref{eq:map} is a contraction mapping with respect to the function $\Psi(\cdot)$ as defined in~\eqref{eq:measure}.
}\end{lem}
\begin{pf}
By~\cite[Proposition 8.3.9]{hernandez2012further}, it suffices to verify that a set of assumptions, i.e., \cite[Assumptions 8.3.1, 8.3.2, and 8.3.3]{hernandez2012further}, are satisfied. Assumptions 8.3.1 and 8.3.3 impose conditions on the compactness of the action space and the continuity of the reward functions in actions, both of which hold naturally in our setup. Furthermore, by \cite[Remark 8.3.5]{hernandez2012further}, Assumption 8.3.2 can be replaced by conditions (I)–(II) provided in Appendix~\ref{apx:stationary}, which have already been verified. Hence, $\mathcal{T}_{\delta}(\cdot)$ is a contraction mapping.
\end{pf}
Furthermore, with (I)-(III) in Appendix~\ref{apx:stationary}, a stationary policy satisfying~\eqref{eq:bellman discount} exists~\cite[Lemma 3.2]{guo2006average}. Thus, we are able to prove the monotonicity of $V_{\delta}(\phi)$ by induction.

Consider a function $\hat{V}_{\delta}(\cdot)$. Assume Lemma~\ref{lemma:monototic} holds for $\hat{V}_{\delta}(\cdot)$, i.e., $\hat{V}_{\delta}(\phi_+)\geq\hat{V}_{\delta}(\phi_-)$ for all $\phi_+$ and $\phi_-$. Let $a^*_-$ denote the action that maximizes the right-hand side of~\eqref{eq:map} for $\phi_-$, i.e., $a_-^*=\arg\max_{a\in\mathbb{A}_{\phi_-}}\{r(\phi_-,a)+\delta\sum_{\phi'\in\mathbb{S}}\mathbb{P}(\phi_-'|\phi_-,a)\hat{V}_{\delta}(\phi_-')\}$. Then, by assumption:
\begin{equation*}
\sum_{\phi'_+\in\mathbb{S}}\mathbb{P}(\phi'_+|\phi_+,a_-^*)\hat{V}_{\delta}(\phi_+')\geq \sum_{\phi'_-\in\mathbb{S}}\mathbb{P}(\phi'_-|\phi_-,a_-^*)\hat{V}_{\delta}(\phi_-').
\end{equation*}
Additionally, since $r(\phi_+,a^*_-)\geq  r(\phi_-,a^*_-)$, we obtain
\begin{equation*}
\begin{aligned}
&\mathcal{T}_{\delta}\hat{V}_{\delta}(\phi_+)	\geq r(\phi_+,a^*_-)+\delta\sum_{\phi'_+\in\mathbb{S}}\mathbb{P}(\phi'_+|\phi_+,a_-^*)\hat{V}_{\delta}(\phi_+')\\
				\geq& r(\phi_-,a^*_-)+\delta\sum_{\phi'_-\in\mathbb{S}}\mathbb{P}(\phi'_-|\phi_-,a_-^*)\hat{V}_{\delta}(\phi_-')=\mathcal{T}_{\delta}\hat{V}_{\delta}(\phi_-).
\end{aligned}
\end{equation*}
Therefore, Lemma~\ref{lemma:monototic} still holds after the Bellman operator $\mathcal{T}_{\delta}(\cdot)$. By Banach fixed-point theorem, since the mapping is contractive, it has a unique fixed point. By induction, the fixed point $V_{\delta}(\cdot)$ satisfies $V_{\delta}(\phi_+)\geq V_{\delta}(\phi_-)$ for all $\phi_+$ and $\phi_-$.

We note that the sequence $\{V_{\delta(n)}(\phi)-V_{\delta(n)}(\phi_f)\}$ is equicontinuous, where $\{\delta(n)\}$ is any sequence of increasing discount factors such that $\delta(n)\to1$ as $n\to\infty$. Furthermore, since conditions (I)-(V) in Appendix~\ref{apx:stationary} are satisfied, it follows from~\cite[Remark 4.1(b)]{guo2006average} that, $V(\phi)=\lim_{n\to\infty}\{V_{\delta(n)}(\phi)-V_{\delta(n)}(\phi_f)\}$. As a result, $V(\phi)$ is also monotonically increasing in $\tau^{(i)}$. This completes the proof.

\section{Proof of Lemma~\ref{lemma:super}}\label{apx:additive}

For notation simplicity, we denote the packet arrival rate of the $i$-th sensor under action $a$ as $f_a^{(i)}\triangleq\lambda^{(i)}=f(\frac{H^{(i)}}{a^{(i)}G^{(i)}+\sigma_i^2})$. We let $\hat{\phi}_+=(b,E,\hat{\Lambda}_+^{(1)},\dots,\hat{\Lambda}_+^{(N)})$ and $\hat{\phi}_-=(b,E,\hat{\Lambda}_-^{(1)},\allowbreak\dots,\hat{\Lambda}_-^{(N)})$, where for $j=1,\dots,N$, $\hat{\Lambda}^{(j)}_+	=(H^{(j)},G^{(j)},\allowbreak\tau^{(j)}_+)$ and $\hat{\Lambda}^{(j)}_- =(H^{(j)},G^{(j)},\tau^{(j)}_-)$. Here, $b$, $E$, $H^{(j)}$ and $G^{(j)}$ are some arbitrary values. Additionally, for the specific index $i$, $\tau^{(i)}_-\leq  \tau^{(i)}_+$, and for $j\neq i$, $\tau^{(j)}_-=  \tau^{(j)}_+=\tau^{(j)}$. 

Then for the $i$-th system, since $f_{a_{-,i}}^{(i)}{\geq f_{a_{+,i}}^{(i)}}$ and $0\leq \tr[h^{\tau^{(i)}_-+1}(\bar{P}_{i})]\leq \tr[h^{\tau^{(i)}_++1}(\bar{P}_{i})]$, we have $(f_{a_{-,i}}^{(i)}- f_{a_{+,i}}^{(i)})\tr[h^{\tau^{(i)}_++1}(\bar{P}_{i})]\geq(f_{a_{-,i}}^{(i)}- f_{a_{+,i}}^{(i)})\tr[h^{\tau^{(i)}_-+1}(\bar{P}_{i})]$.

For $j\neq i$, $(f_{a_{-,i}}^{(j)}- f_{a_{+,i}}^{(j)})\tr[h^{\tau^{(j)}_++1}(\bar{P}_{j})]=(f_{a_{-,i}}^{(j)}- f_{a_{+,i}}^{(j)})\tr[h^{\tau^{(j)}_-+1}(\bar{P}_{j})]=(f_{a_{-,i}}^{(j)}- f_{a_{+,i}}^{(j)})\tr[h^{\tau^{(j)}_-+1}(\bar{P}_{j})]$.

On the other hand, applying {Lemma}~\ref{lemma:monototic} and following a similar trend as above , we have $\mathbb{E}[V(\phi')|\hat{\phi}_+,a_+]-\mathbb{E}[V({\phi}')|\hat{\phi}_+,a_-]\geq\mathbb{E}[V(\phi')|\hat{\phi}_-,a_+]-\mathbb{E}[V(\phi')|\hat{\phi}_-,a_-]$, where $\mathbb{E}[V(\phi')|\phi,a]=\sum_{\phi'\in\mathbb{S}}\mathbb{P}(\phi'|\phi,a)V(\phi')$.

Combining the inequalities above, we obtain{\small
\begin{equation*}
\begin{aligned}
	& Q (\tau^{(i)}_+,a_{+,i})- Q (\tau^{(i)}_+,a_{-,i})\\
=	& \sum_{j=1}^N (f_{a_{+,i}}^{(j)}-f_{a_{-,i}}^{(j)})\tr[\bar{P}_{j}]+(f_{a_{-,i}}^{(j)}- f_{a_{+,i}}^{(j)})\tr[h^{\tau^{(j)}_++1}(\bar{P}_{j})]\\
	&\quad+	\mathbb{E}[V(\phi')|\hat{\phi}_+,a_+]-\mathbb{E}[V(\phi')|\hat{\phi}_+,a_-]\\
\geq	&\sum_{j=1}^N (f_{a_{+,i}}^{(j)}-f_{a_{-,i}}^{(j)})\tr[\bar{P}_{j}]+(f_{a_{-,i}}^{(j)}- f_{a_{+,i}}^{(j)})\tr[h^{\tau^{(j)}_-+1}(\bar{P}_{j})]\\
	&\quad+\mathbb{E}[V(\phi')|\hat{\phi}_-,a_+]-\mathbb{E}[V(\phi')|\hat{\phi}_-,a_-]\\
=	& Q (\tau^{(i)}_-,a_{+,i})- Q (\tau^{(i)}_-,a_{-,i}),
\end{aligned}
\end{equation*}}
which completes the proof.

\end{document}